\newif\ifGross

\Grosstrue
\Grossfalse

\ifGross
\documentclass[12pt]{article}
\usepackage{fullpage}
\else
\documentclass[12pt]{article}
\usepackage{fullpage}
\fi

\usepackage[color]{changebar}
\usepackage{amsfonts}
\usepackage{amsmath,amssymb}
\usepackage{bm}
\usepackage{mathtools}
\usepackage{graphicx}
\usepackage[ruled,linesnumbered,vlined]{algorithm2e}
\usepackage{multirow}
\usepackage{mathptmx}      
\usepackage{algorithm2e}

\usepackage{xspace}
\usepackage{makecell}
\usepackage{doi}
\usepackage{todonotes}

\usepackage{epstopdf} 
\usepackage{verbatim}
\usepackage{amsthm}
\usepackage{xspace}

\usepackage{dsfont}

\usepackage{pgfplots}
\pgfplotsset{compat=1.14}
\usepackage{subfig}

\usepackage{array}
\usepackage{tabu}
\usepackage{mathtools}
\usepackage{nicefrac}

\usepackage{tabularx,booktabs}

\usepackage{fix-cm}

\usepackage{wrapfig}
\newcolumntype{Y}{>{\centering\arraybackslash}X}

\newcommand{\N}{{\mathbb{N}}}
\newcommand{\E}{{\rm E}}
\newcommand{\Var}{{\rm Var}}

\providecommand{\stirling}[2]{ {\genfrac{[}{]}{0pt}{}{#1}{#2}}}

\graphicspath{{./Figs/}}

\allowdisplaybreaks 

\newtheorem{theorem}{Theorem}
\newtheorem{lemma}[theorem]{Lemma}
\newtheorem{corollary}[theorem]{Corollary}

\newtheorem{definition}[theorem]{Definition}
\newtheorem{remark}[theorem]{Remark}
\newtheorem{conjecture}[theorem]{Conjecture}

\usepackage{tikz}
\usetikzlibrary{shapes, arrows, positioning, backgrounds,calc,trees,fit,patterns}
\usetikzlibrary{decorations, decorations.text}
\usetikzlibrary{matrix,decorations.pathreplacing}
\usetikzlibrary{arrows.meta,backgrounds,shadings,shadows}

\newcommand{\neighborhood}{\ensuremath{\mathcal{N}}}
\newcommand{\dist}{\ensuremath{d}\xspace}
\newcommand{\onepso}{\textsc{OnePSO}\xspace}
\newcommand{\dpso}{\textsc{D-PSO}\xspace}
\newcommand{\oneea}{\ensuremath{(1+1)}-EA\xspace}
\newcommand{\uar}{u.a.r.\xspace}

\newcommand{\THETA}{\ensuremath{\operatorname{\Theta}}\xspace}
\newcommand{\OMEGA}{\ensuremath{\operatorname{\Omega}}\xspace}
\newcommand{\nat}{\ensuremath{\mathbb{N}}\xspace}

\newcommand{\model}{\ensuremath{\operatorname{\mathcal{M}}\xspace}}
\newcommand{\onemax}{\textsc{OneMax}\xspace} 
\newcommand{\algo}{\ensuremath{\mathcal{A}}\xspace}
\newcommand{\1}{\ensuremath{\mathds{1}}\xspace}

\newcommand{\real}{\ensuremath{\mathbb{R}}\xspace}
\newcommand{\naturalNumbers}{\ensuremath{\mathbb{N}}\xspace}
\newcommand{\R}{\real}

\newcommand{\hoehe}{\vrule height2.8ex depth1.5ex width 0pt}
\newcommand{\hoeheX}{\vrule height3.8ex depth1.5ex width 0pt}

\usepackage{hyperref}

\allowdisplaybreaks 

\begin{document}

\title{Exact Markov Chain-based Runtime Analysis of a Discrete Particle Swarm Optimization Algorithm on Sorting and OneMax\thanks{
    This is the preprint of \cite{MRSW:21} published in the Journal \emph{Natural Computing} and it is the significantly
extended version of \cite{MRSSW:17} published in the \emph{proceedings
of the 14th ACM/SIGEVO Workshop on Foundations of Genetic Algorithms (FOGA)},
2017.}}

\author{Moritz M\"uhlenthaler\\
Laboratoire G-SCOP,\\Grenoble INP, France\\
{\sf\small moritz.muhlenthaler}\\[-2mm] {\sf\small @grenoble-inp.fr}
\and Alexander Ra\ss
\thanks{Corresponding author} \quad Manuel Schmitt\quad Rolf Wanka
\\
Department of Computer Science\\
University of Erlangen-Nuremberg, Germany\\
{\sf\small $\{$alexander.rass, manuel.schmitt, rolf.wanka$\}$@fau.de}
}

\date{ }
\maketitle

\newcommand{\sep}{ $\cdot$ }

\begin{abstract}

Meta-heuristics are powerful tools for solving optimization problems whose
structural properties are unknown or cannot be exploited algorithmically. We
propose such a meta-heuristic for a large class of optimization problems over
discrete domains based on the \emph{particle swarm optimization} (PSO)
paradigm. 
We provide a comprehensive formal analysis of the performance of this algorithm
on certain ``easy'' reference problems in a black-box setting, namely the
sorting problem and the problem \onemax.
In our analysis we use a Markov model of the proposed algorithm to obtain upper
and lower bounds on its expected optimization time.
Our bounds are essentially tight with respect to the Markov model. 
We show that for a suitable choice of algorithm parameters the expected
optimization time is comparable to that of known algorithms and, furthermore, for
other parameter regimes, the algorithm behaves less greedy and more
explorative, which can be desirable in practice in order to escape local optima. 
Our analysis provides a precise insight on the tradeoff between optimization
time and exploration. 
To obtain our results we introduce the notion of \emph{indistinguishability}
of states of a Markov chain and provide bounds on the solution
of a recurrence equation with non-constant coefficients by integration.

\medskip

\end{abstract}

\noindent
This research did not receive any specific grant from funding agencies in the public, commercial, or not-for-profit sectors.

\section{Introduction}
\label{sec:intro}

Meta-heuristics are very successful at finding good solutions
for hard optimization problems in practice. However, due to the nature of such
algorithms and the problems they are applied to, it is generally very difficult
to derive performance guarantees, or to determine the number of steps it takes
until an optimal solution is found. In the present work we propose a simple
adaptation of the particle swarm optimization (PSO) algorithm
introduced by~\cite{eb_ken_1995,ken_eb_1995} to optimization problems over discrete
domains. Our proposed algorithm assumes very little about the problem structure
and consequently, it works naturally for a large class of discrete domains. It is reasonable to expect from a meta-heuristic that it solves
black-box versions of many tractable problems in expected polynomial time. 
We provide a formal analysis based on Markov chains and establish under
which conditions our algorithm satisfies this basic requirement. More concretely,  we
consider two classical problems that are easy to solve in a non-black-box
setting, namely the problem of sorting items by transpositions and the
problem \onemax, which asks to maximize the number of ones in a bitstring. Our
analysis gives precise information about the expected number of steps our
algorithm takes in order to solve these two reference problems. 
Our runtime
bounds are essentially tight with respect to the Markov process we use to model
the behavior of the algorithm.

For practical purposes, 
a meta-heuristic should, in one way or another, incorporate the following two
general strategies: i) find an improving solution locally 
(often referred to as \emph{exploitation})  and ii) move to unexplored parts of
the search space (often referred to as \emph{exploration}).
The first strategy essentially leads the algorithm to a local optimum while the
second one helps the algorithm to avoid getting stuck when it is close to a local optimum. For our proposed algorithm, as for many other meta-heuristics,
the tradeoff between the
two strategies can be conveniently set by an algorithm parameter. Our analysis
shows that there is a sharp threshold with respect to this parameter, where the
expected runtime of the algorithm on the reference problems turns from
polynomial to exponential. Hence, we can maximize the algorithm's ability to
escape local optima while still maintaining polynomial runtime on the reference
problems.

A key tool for the runtime analysis of meta-heuristics for optimization
problems over discrete domains is the \emph{fitness level method} pioneered by
\cite{Wegener:02}. The basic idea is to consider the level sets of the
objective function of a problem instance  and to determine the expected number
of steps an algorithm takes
to move to a better level set. This approach has been used extensively in the
study of so-called \emph{elitist $(1+1)$-EAs} \cite{Wegener:02,DJW:02,GW:03,S:13}.
These algorithms keep a single ``current'' solution and update this solution
only if a better one is found. 
Our analysis of the proposed PSO algorithm also relies on the fitness level
method. However, since our algorithm also considers non-improving solutions in
order to escape local optima, a much more involved analysis is required in
order to determine the time it takes to move to a better fitness level.

We will refer in the following by \emph{expected optimization time} to the
\emph{expected} number of evaluations of the objective function an algorithm
performs until an optimal solution is found. Before giving a precise statement
of our results we provide some background information on the PSO algorithm as well as the
two reference problems we consider.

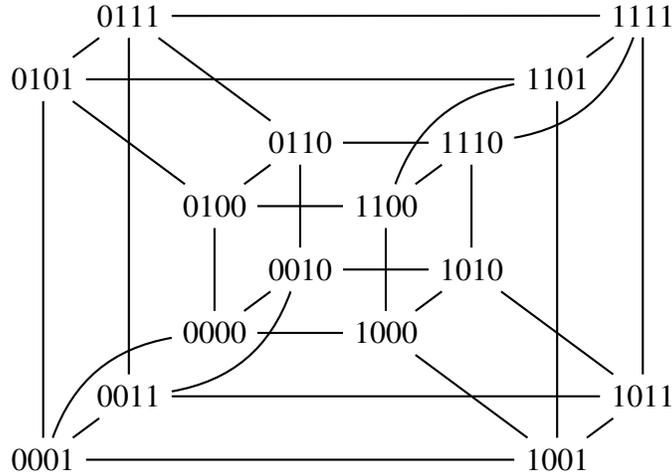
\begin{figure}[t]
\begin{center}
\begin{tikzpicture}
    \matrix[matrix of math nodes,column sep={2.7em,between origins},row sep={2.0em,between origins}]{
                 &|(0111)| 0111&             &             &             &             &             &|(1111)| 1111\\
    |(0101)| 0101&             &             &             &             &             &|(1101)| 1101&             \\
                 &             &             &|(0110)| 0110&             &|(1110)| 1110&             &             \\
                 &             &|(0100)| 0100&             &|(1100)| 1100&             &             &             \\
                 &             &             &|(0010)| 0010&             &|(1010)| 1010&             &             \\
                 &             &|(0000)| 0000&             &|(1000)| 1000&             &             &             \\
                 &|(0011)| 0011&             &             &             &             &             &|(1011)| 1011\\
    |(0001)| 0001&             &             &             &             &             &|(1001)| 1001&             \\
    };
	\path[thick](0000) edge [bend right=30] (0001);
    \draw[thick](0000) -- (0010);
    \draw[thick](0000) -- (0100);
    \draw[thick](0000) -- (1000);
    \draw[thick](0001) -- (0011);
    \draw[thick](0001) -- (0101);
    \draw[thick](0001) -- (1001);
	\path[thick](0010) edge [bend left=30 ] (0011);
    \draw[thick](0010) -- (0110);
    \draw[thick](0010) -- (1010);
    \draw[thick](0011) -- (0111);
    \draw[thick](0011) -- (1011);
    \draw[thick](0100) -- (0101);
    \draw[thick](0100) -- (0110);
    \draw[thick](0100) -- (1100);
    \draw[thick](0101) -- (0111);
    \draw[thick](0101) -- (1101);
    \draw[thick](0110) -- (0111);
    \draw[thick](0110) -- (1110);
    \draw[thick](0111) -- (1111);
    \draw[thick](1000) -- (1001);
    \draw[thick](1000) -- (1010);
    \draw[thick](1000) -- (1100);
    \draw[thick](1001) -- (1011);
    \draw[thick](1001) -- (1101);
    \draw[thick](1010) -- (1011);
    \draw[thick](1010) -- (1110);
    \draw[thick](1011) -- (1111);
    \draw[thick](1100) -- (1110);
	\path[thick](1100) edge [bend left=30 ] (1101);
    \draw[thick](1101) -- (1111);
	\path[thick](1110) edge [bend right=30] (1111);
    ;
\end{tikzpicture}
\end{center}
  \caption{Search space for the problem of \onemax on $\lbrace 0,1\rbrace^4$ by
  bitflips, i.\,e., the $4$-dimensional hypercube. Two bitstrings $x,y$ are
  adjacent iff $x$ and $y$ differ in exactly one position.}
  \label{fig:hypercube}
\end{figure}

\subsection{Particle Swarm Optimization}

The PSO algorithm has been introduced by~\cite{eb_ken_1995,ken_eb_1995} and is inspired by the social
interaction of bird flocks. Fields of successful application of PSO are, among
many others, Biomedical Image Processing~\cite{SSW:15,WSZZE:04},
Geosciences~\cite{OD:10}, Agriculture~\cite{YWetal:17}, and Materials
Science~\cite{RPPN:09}. In the continuous setting, it is known that the
algorithm converges to a local optimum under mild
assumptions~\cite{SW:15}. The algorithm has been adapted to various
discrete problems and several results are available, for instance for
binary problems~\cite{SW:10} and the traveling salesperson problem
(TSP)~\cite{HMHW:11}.

A PSO algorithm manages a collection (called \emph{swarm}) of particles. Each
particle consists of an (admissible) solution together with a velocity vector.
Additionally each individual particle knows the \emph{local attractor}, which
is the best solution found by that particle.  Information between particles is
shared via a common reference solution called \emph{global attractor}, which is
the best solution found so far by all particles. In each iteration of the
algorithm, the solution of each particle is updated based on its relative
position with respect to the attractors and some random perturbation. Algorithm
parameters balance the influence of the attractors and the perturbation and hence
give a tradeoff between the two general search strategies ``exploration'' and
``exploitation''.  Although PSO has originally been proposed to solve
optimization problems over a --- typically rectangular ---  domain $X \subseteq
\real^n$,
several authors have adapted PSO to discrete domains. This 
requires a fundamental reinterpretation of the PSO movement equation because
corresponding mathematical operations of a vector space are typically lacking
in the discrete setting.
An early discrete PSO variant is the \emph{binary PSO} presented in \cite{KE:97} for optimizing over
$X=\{0,1\}^n$ where velocities determine probabilities such that a
bit is zero or one in the next iteration. A PSO variant for 
optimizing over general integral domains $X=\{0,1,\ldots,M-1\}^n$, $M \in \nat$, has been proposed in \cite{VO:07}.

\subsection{Problems and Search Spaces}
In this section we briefly define the optimization problems for which we will
study the performance of the proposed PSO algorithm.  The problem \onemax asks
for a binary string of length $n$ that maximizes the function
\[
  \onemax(x_1,\ldots,x_n) = \smash[t]{\sum_{i=1}^n x_i}\enspace,
\]
which counts the number of ones in a binary string. A more general version of
this problem asks to minimize the Hamming distance to an unknown binary string
of length $n$. The proposed algorithm works exactly the same on the more
general problem since it is indifferent to the actual bit values and each bit is
handled independently.  Therefore, the performance of our algorithm on this more
general version is equal to its performance on \onemax.
The corresponding search space is the $n$-dimensional hypercube: Any binary
string of length $n$ is a (feasible) solution, and two solutions are adjacent
iff they differ by exactly one bitflip. For $n=4$, the search space is shown in
Figure~\ref{fig:hypercube}. More generally, a \emph{pseudo-Boolean function} is
any function $f : \{0, 1\}^n \to \real$.

By the \emph{sorting problem} we refer to the task of arranging $n$ items in
non-decreasing order using transpositions.  An (algebraic) transposition
$t=(i\,\,\, j)$ is the exchange of the entries at positions $i$ and $j$.
Therefore, the search space is the following (undirected) graph: The vertices
are the permutations on $\{1,2,\ldots,n\}$ and two vertices $x, y$ are adjacent
iff there is a transposition $t$ such that $x \circ t = y$. The objective
function is the transposition distance to the identity
permutation\footnote{Note that a different definition of ``transposition'' is
used in computational biology, so ``transposition distance'' has a
different meaning, e.g., in~\cite{BP:98}.}. Figure~\ref{fig:permutations} shows the
search space for the problem of sorting items $\{1, 2, 3, 4\}$ using
transpositions. Any two permutations drawn in the same vertical layer have the
same objective value.
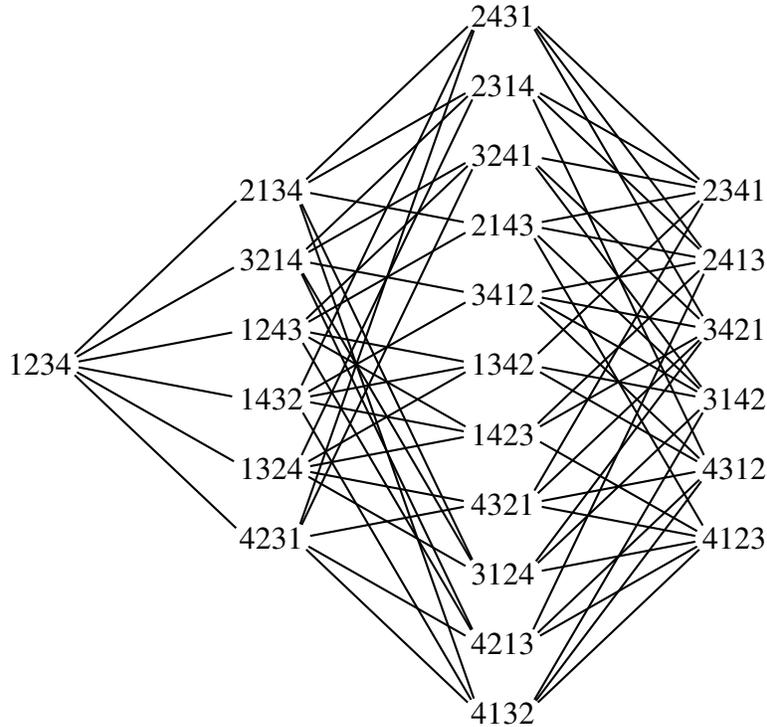
\begin{figure}[t]
\centering
    \begin{tikzpicture}
    \matrix[matrix of math nodes,column sep={7.3em,between origins},row sep={1.1em,between origins}]{
                   &               & |(2431)| 2431 &               \\
                   &               &               &               \\
                   &               & |(2314)| 2314 &               \\
                   &               &               &               \\
                   &               & |(3241)| 3241 &               \\
                   & |(2134)| 2134 &               & |(2341)| 2341 \\
                   &               & |(2143)| 2143 &               \\
                   & |(3214)| 3214 &               & |(2413)| 2413 \\
                   &               & |(3412)| 3412 &               \\
                   & |(1243)| 1243 &               & |(3421)| 3421 \\
     |(1234)| 1234 &               & |(1342)| 1342 &               \\
                   & |(1432)| 1432 &               & |(3142)| 3142 \\
                   &               & |(1423)| 1423 &               \\
                   & |(1324)| 1324 &               & |(4312)| 4312 \\
                   &               & |(4321)| 4321 &               \\
                   & |(4231)| 4231 &               & |(4123)| 4123 \\
                   &               & |(3124)| 3124 &               \\
                   &               &               &               \\
                   &               & |(4213)| 4213 &               \\
                   &               &               &               \\
                   &               & |(4132)| 4132 &               \\
    };
    \node[circle] at ([xshift=-0.6em]1234.east) (1234E){};
    \node[circle] at ([xshift=-0.6em]2134.east) (2134E){};
    \node[circle] at ([xshift=-0.6em]3214.east) (3214E){};
    \node[circle] at ([xshift=-0.6em]1243.east) (1243E){};
    \node[circle] at ([xshift=-0.6em]1432.east) (1432E){};
    \node[circle] at ([xshift=-0.6em]1324.east) (1324E){};
    \node[circle] at ([xshift=-0.6em]4231.east) (4231E){};
    \node[circle] at ([xshift=-0.6em]2431.east) (2431E){};
    \node[circle] at ([xshift=-0.6em]2314.east) (2314E){};
    \node[circle] at ([xshift=-0.6em]3241.east) (3241E){};
    \node[circle] at ([xshift=-0.6em]2143.east) (2143E){};
    \node[circle] at ([xshift=-0.6em]3412.east) (3412E){};
    \node[circle] at ([xshift=-0.6em]1342.east) (1342E){};
    \node[circle] at ([xshift=-0.6em]1423.east) (1423E){};
    \node[circle] at ([xshift=-0.6em]4321.east) (4321E){};
    \node[circle] at ([xshift=-0.6em]3124.east) (3124E){};
    \node[circle] at ([xshift=-0.6em]4213.east) (4213E){};
    \node[circle] at ([xshift=-0.6em]4132.east) (4132E){};
    \node[circle] at ([xshift= 0.6em]2134.west) (2134W){};
    \node[circle] at ([xshift= 0.6em]3214.west) (3214W){};
    \node[circle] at ([xshift= 0.6em]1243.west) (1243W){};
    \node[circle] at ([xshift= 0.6em]1432.west) (1432W){};
    \node[circle] at ([xshift= 0.6em]1324.west) (1324W){};
    \node[circle] at ([xshift= 0.6em]4231.west) (4231W){};
    \node[circle] at ([xshift= 0.6em]2431.west) (2431W){};
    \node[circle] at ([xshift= 0.6em]2314.west) (2314W){};
    \node[circle] at ([xshift= 0.6em]3241.west) (3241W){};
    \node[circle] at ([xshift= 0.6em]2143.west) (2143W){};
    \node[circle] at ([xshift= 0.6em]3412.west) (3412W){};
    \node[circle] at ([xshift= 0.6em]1342.west) (1342W){};
    \node[circle] at ([xshift= 0.6em]1423.west) (1423W){};
    \node[circle] at ([xshift= 0.6em]4321.west) (4321W){};
    \node[circle] at ([xshift= 0.6em]3124.west) (3124W){};
    \node[circle] at ([xshift= 0.6em]4213.west) (4213W){};
    \node[circle] at ([xshift= 0.6em]4132.west) (4132W){};
    \node[circle] at ([xshift= 0.6em]2341.west) (2341W){};
    \node[circle] at ([xshift= 0.6em]2413.west) (2413W){};
    \node[circle] at ([xshift= 0.6em]3421.west) (3421W){};
    \node[circle] at ([xshift= 0.6em]3142.west) (3142W){};
    \node[circle] at ([xshift= 0.6em]4312.west) (4312W){};
    \node[circle] at ([xshift= 0.6em]4123.west) (4123W){};
    \draw[thick](1234E) -- (2134W);
    \draw[thick](1234E) -- (3214W);
    \draw[thick](1234E) -- (1324W);
    \draw[thick](1234E) -- (4231W);
    \draw[thick](1234E) -- (1432W);
    \draw[thick](1234E) -- (1243W);
    \draw[thick](2134E) -- (3124W);
    \draw[thick](2134E) -- (2314W);
    \draw[thick](2134E) -- (4132W);
    \draw[thick](2134E) -- (2431W);
    \draw[thick](2134E) -- (2143W);
    \draw[thick](3214E) -- (2314W);
    \draw[thick](3214E) -- (3124W);
    \draw[thick](3214E) -- (4213W);
    \draw[thick](3214E) -- (3412W);
    \draw[thick](3214E) -- (3241W);
    \draw[thick](1324E) -- (3124W);
    \draw[thick](1324E) -- (2314W);
    \draw[thick](1324E) -- (4321W);
    \draw[thick](1324E) -- (1423W);
    \draw[thick](1324E) -- (1342W);
    \draw[thick](4231E) -- (2431W);
    \draw[thick](4231E) -- (3241W);
    \draw[thick](4231E) -- (4321W);
    \draw[thick](4231E) -- (4132W);
    \draw[thick](4231E) -- (4213W);
    \draw[thick](1432E) -- (4132W);
    \draw[thick](1432E) -- (3412W);
    \draw[thick](1432E) -- (1342W);
    \draw[thick](1432E) -- (2431W);
    \draw[thick](1432E) -- (1423W);
    \draw[thick](1243E) -- (2143W);
    \draw[thick](1243E) -- (4213W);
    \draw[thick](1243E) -- (1423W);
    \draw[thick](1243E) -- (3241W);
    \draw[thick](1243E) -- (1342W);
    \draw[thick](2143E) -- (4123W);
    \draw[thick](2143E) -- (2413W);
    \draw[thick](2143E) -- (3142W);
    \draw[thick](2143E) -- (2341W);
    \draw[thick](4213E) -- (2413W);
    \draw[thick](4213E) -- (4123W);
    \draw[thick](4213E) -- (4312W);
    \draw[thick](1423E) -- (4123W);
    \draw[thick](1423E) -- (2413W);
    \draw[thick](1423E) -- (3421W);
    \draw[thick](3241E) -- (2341W);
    \draw[thick](3241E) -- (3421W);
    \draw[thick](3241E) -- (3142W);
    \draw[thick](1342E) -- (3142W);
    \draw[thick](1342E) -- (4312W);
    \draw[thick](1342E) -- (2341W);
    \draw[thick](3124E) -- (4123W);
    \draw[thick](3124E) -- (3421W);
    \draw[thick](3124E) -- (3142W);
    \draw[thick](2314E) -- (4312W);
    \draw[thick](2314E) -- (2413W);
    \draw[thick](2314E) -- (2341W);
    \draw[thick](4321E) -- (3421W);
    \draw[thick](4321E) -- (2341W);
    \draw[thick](4321E) -- (4123W);
    \draw[thick](4321E) -- (4312W);
    \draw[thick](4132E) -- (3142W);
    \draw[thick](4132E) -- (4312W);
    \draw[thick](4132E) -- (4123W);
    \draw[thick](3412E) -- (4312W);
    \draw[thick](3412E) -- (3142W);
    \draw[thick](3412E) -- (2413W);
    \draw[thick](3412E) -- (3421W);
    \draw[thick](2431E) -- (3421W);
    \draw[thick](2431E) -- (2341W);
    \draw[thick](2431E) -- (2413W);
\end{tikzpicture}
  \caption{Search space for the problem of sorting four items by transpositions. Two permutations
	$x,y$ on $\{1,2,3,4\}$ are adjacent iff there is a transposition $t$ such that $x \circ t = y$.}
  \label{fig:permutations}
\end{figure}

The sorting problem and \onemax have a unique optimum.  Furthermore, the value
of the objective function is the distance to the unique optimal solution in the
corresponding directed graph.

\subsection{Our Contribution}

We propose a simple adaptation of the PSO algorithm to optimization problems
over discrete domains. We refer to this algorithm as \dpso. The algorithm works
naturally on a large class of discrete problems, for instance optimization
problems over bitstrings, integral domains, and permutations. The general task is to optimize
a function $f:X\rightarrow\real$, where $X$ is a finite set of feasible
solutions.
Our assumptions on the problem structure are the following. We assume that the
set $X$ is the vertex set of a finite, strongly connected
graph and for any solution $x \in X$, we can sample a neighbor of $x$ efficiently and
uniformly. The \dpso algorithm essentially explores this graph, looking for an optimal vertex.
In our analysis, we assume at first a swarm size of one as in~\cite{MRSSW:17},
similar to the analysis of EAs and ACO in~\cite{SW:10}. We refer to the
corresponding specialization of \dpso as \onepso.
Indeed, for a single particle we have only a single attractor and,
as a consequence, a single parameter is sufficient to control
the tradeoff between the moving towards the attractor
and performing a random perturbation.

\begin{table*}[t] 
  \caption{ Summary of upper and lower bounds on the expected time
  taken by the algorithm \onepso to solve the sorting problem and \onemax for
  $c \in [0,1]$. The functions $\alpha(c)$ and $\beta(c)$ are given in Lemma
  \ref{lemma:lower_bound_sort_small_c} and Lemma
  \ref{lemma:lower_bound_1max_small_c}, respectively. Note that
  $1<\beta(c)<2$ and $\beta(c)<\alpha(c)<3+2\cdot \sqrt{2}<6$ for all $c\in(0,1/2)$.}

  \centering
  \begin{tabu} to 0.95\linewidth {l|X|X||X|X}
	&	\multicolumn{2}{c||}{\hoehe sorting}													&	\multicolumn{2}{c}{\onemax}	\tabularnewline\cline{2-3}\cline{4-5}
		&	\hoehe lower bound						&	upper bound		                                           	&	lower bound					&	upper bound	\tabularnewline\hline
		\hoeheX $c = 1$					& $\Omega(n^2)$		&	$O(n^2 \log n)$								& \multicolumn{2}{c}{$\Theta(n \log n)$} 						\tabularnewline
		\hoehe $c \in (\frac{1}{2},1)$	& $\Omega(n^2)$						&	$O(n^2 \log n)$												& \multicolumn{2}{c}{$\Theta(n \log n)$}  	\tabularnewline
		\hoehe $c = \frac{1}{2}$		& $\Omega(n^\frac{8}{3})$			&	$O(n^3 \log n)$												& $\Omega(n^\frac{3}{2})$		&	$O(n^{\frac{3}{2}} \log n)$							  	\tabularnewline	
		\hoehe $c \in (0,\frac{1}{2})$	& $\Omega(\alpha(c)^n \cdot n^2)$	&	$O\left( \left(\frac{1-c}{c}\right)^n n^2 \log n\right)$	& $\Omega(\beta(c)^n \cdot n)$	&	$O\left( \beta(c)^n \cdot n^2 \log n\right)$	\tabularnewline		
		\hoehe $c = 0$					& \multicolumn{2}{c||}{$\Theta(n!)^\dag$}																	& \multicolumn{2}{c}{$\Theta(2^n)$}	
  \end{tabu}
  \\$^\dag$ {The upper bound $O(n!)$ is conjectured. All other bounds --- including the lower bound $\Omega(n!)$ --- are proved formally in this work.}

  \label{tab:summary}
\end{table*}

Our main results are upper and lower bounds on the expected optimization time
of the proposed \onepso algorithm for solving the sorting problem and \onemax
in a black-box setting summarized in Table~\ref{tab:summary}. Certainly there are
faster algorithms for the sorting problem or \onemax in a non-black-box setting, e.\,g., quicksort for
the sorting problem.
The upper bounds we prove for \onepso naturally hold for \dpso and the bounds are
tight with respect to our Markov model.
The algorithm parameter $c$ determines the probability
of making a move towards the attractor. Depending on the parameter $c$, we
obtain a complete classification of the expected optimization time of \onepso.
For $c=0$, \onepso
performs a random walk on the search space, and for $c=1$, the algorithm
behaves like randomized local search (see~\cite{P:90} or \cite{DoerrNeumann2020} for results on local search variants).
For $c\in(1/2,1)$ the \onepso behaves essentially like the \oneea
variants from~\cite{DJW:02,STW:04}, since \oneea variants perform in expectation a
constant number of elementary mutations to obtain an improved solution
and \onepso with $c\in(1/2,1)$ in expectation also performs a
constant number of elementary mutations to find an improved solution, before
it returns to the current best solution. Therefore, \onepso in
a sense generalizes the \oneea algorithm since a parameter choice in
$c\in[0,1]$ supplies a broader range of behavior options than exploring
solutions which are in expectation a constant number of elementary mutations
away from the current best solution. If $c<1$ then the \onepso uses similar
subroutines as the Metropolis algorithm (see~\cite{M:53}), but for \onepso new
positions are always accepted and guidance to good solutions is instead
implemented by a ``drift'' to the best position found so far.

Indeed bounds on the expected optimization time for \onemax and upper bounds on
the expected optimization time for the sorting problem of the \onepso with
parameter $c\in(1/2,1]$ match the respective bounds on the expected
optimization time for \oneea variants from~\cite{DJW:02,STW:04}.
We show that
for $c  \in [1/2, 1)$, the expected optimization time of \onepso for sorting and
\onemax is polynomial and for $c \in (0, 1/2)$, the expected optimization time
is exponential. For $c$ in the latter range we provide lower bounds on the base
of the exponential expression by $\alpha(c)$ for sorting and $\beta(c)$ for
\onemax such that $1 < \beta(c) < \alpha(c) < 6$ (see
Figure~\ref{fig:alphabeta}). Please note that $\alpha$ and $\beta$ have been
significantly improved compared to the conference version \cite{MRSSW:17} and
the upper bound on \onemax has also been reduced heavily to an exponential term
with base $\beta(c)$. This means that the lower and upper bound on the base of
the exponential expression for the expected optimization time is equal. 
Note that for $c=1/2$ the expected time it takes to visit the attractor again
after moving away from it is maximal while keeping the expected optimization
time polynomial, i.\,e., we have a phase transition at $c=1/2$ such that for
any $c\leq 1/2$ the expected optimization time is polynomial and for any
$c>1/2$ the expected optimization time is exponential in $n$.
Hence, the parameter choice $c=1/2$ maximizes the time until the attractor is visited
again, i.\,e., the particles can explore the search space to the largest
possible extent, provided that \onemax and the sorting problem are solved
efficiently in a black-box setting.

In order to obtain the bounds shown in Table~\ref{tab:summary}, we use a
Markov model which captures the behavior of the \onepso algorithm between two
consecutive updates of the attractor.
Depending on whether we derive upper or lower bounds on the expected optimization time, the Markov model is
instantiated in a slightly different way. The relevant quantity we extract from
the Markov model is the expected number of steps it takes until the \onepso
algorithm returns to the attractor. We determine $\Theta$-bounds on the
expected return time by an analysis of appropriate recurrence equations. 
Similar recurrences occur, for example, in the runtime analysis of randomized
algorithms for the satisfiability problem~\cite{P:91,S:99}. Thus, our analysis
of the Markov model presented in Section~\ref{sec:model} may be of independent
interest.
For $c > 1/2$, the recurrence equations can be solved using standard methods.
For the parameter choice $c \leq  1/2$ however, 
we need to solve recurrence equations with
non-constant coefficients in order to get sufficiently accurate bounds from the
model.
The gaps between upper and lower bounds on the expected optimization times
shown in Table~\ref{tab:summary} result from choosing best-case or worst-case
bounds on the transition probabilities in the Markov model, which are specific
to the optimization problem. Since our bounds on the transition probabilities
are essentially tight, we can hope to close the gap between the upper and lower
bounds on the sorting problem (especially in the exponential case) only by
using a more elaborate model. 

Furthermore, based on Wald's equation and the Blackwell-Girshick equation  we
obtain also the \emph{variance} of the number of function evaluations needed to
find an optimal solution with respect to the Markov model.

\emph{Upper Bounds.} To obtain the upper bounds shown in Table~\ref{tab:summary}
we use the established fitness level method (e.\,g., see
\cite{Wegener:02}). We instantiate our Markov model such that improvements of
the attractor are only accounted for if the current position is at the
attractor. The main difficulty is to determine the expected number of steps
needed to return to the attractor. We obtain this quantity from the analysis of
the corresponding recurrences with constant and non-constant coefficients. 
Furthermore, we obtain by integration closed-form expressions of the expected
number of steps it takes to return to the attractor after an unsuccessful
attempt to improve the current best solution.

\emph{Lower Bounds.} The runtime of the \onepso algorithm is dominated by the time required
for the last improvement of the attractor, after which the global optimum has
been found. We again use the Markov model and observe that in this situation,
the global optimum can be reached only when the Markov model is in a specific
state. We argue that the optimal solution is included in a certain set
$\hat{Y}$ of indistinguishable states.  Therefore, in expectation, this
set needs to be hit $\Omega(\vert\hat{Y}\vert)$ times until the optimum
has been found. By evaluation of the return time to the attractor we also
obtain bounds on the return time to the set $\hat{Y}$.
Furthermore, for \dpso with a constant number of particles, we give a lower
bound of $\Omega(\log n)$ for optimizing a pseudo-Boolean function and for $P =
\rm poly(n)$ particles a stronger lower bound of $\Omega(P\cdot n)$ for the
same task.

\emph{Open Problems.}
Finally, we conjecture that the expected optimization time of \onepso for
sorting $n$ items is asymptotically equivalent to $n!$ if the attractor is not
used at all ($c = 0$). An equivalent statement is that a random walk on the set of
permutations of $n$ items using single transpositions as neighborhood relation
asymptotically takes expected time $n!$ to discover a fixed given
permutation starting at a random position. We provide theoretical evidence for this
conjecture in Appendix~\ref{subsec:randomwalkconjecture}. Furthermore, we
conjecture stronger lower bounds for \onepso for sorting $n$ items for $c > 0$
and provide evidence in Appendix~\ref{subsec:approxOptTime}.

\emph{Extensions.} This article extends the conference paper \cite{MRSSW:17} as
follows. We present \dpso, a discrete PSO algorithm with multiple particles in
Section~\ref{sec:dpso} and show in Section~\ref{sec:dpsobounds} which results
for \onepso generalize to \dpso. For Pseudo-Boolean functions, a new general
lower bound is presented, which holds for \dpso (see
Subsection~\ref{sec:pseudoboolean:lb}).  Furthermore, we give a refined
analysis in the case of exponential runtime which allows us to determine the
exact base of the exponential expression (see
Section~\ref{subsec:integrationbounds}).  In addition to the analysis of the
expected optimization time, we also consider the variance of the expected
optimization time (see Subsection~\ref{subsec:varianceAnalysis}).  We
conjecture that the expected number of steps needed until the random walk on
the space of permutations hits a fixed permutation is $n!$. We provide evidence
for this conjecture in Appendix~\ref{subsec:randomwalkconjecture}.
Finally, we present an approximate model of the algorithm with
average transition
probabilities to obtain the actual bounds where lower and upper bounds on the
expected optimization time differ (see Appendix~\ref{subsec:approxOptTime}).
Also, some theorems have been extended to larger classes of functions, which
may make them more useful in other settings.

\subsection{Related Work}
\label{sec:runtime:related}

Runtime results are available for several other meta-heuristics for
optimization problems over discrete domains, for example evolutionary
algorithms (EAs)~\cite{DJW:02,GW:03,Wegener:02,ADFH:2018,ADY:2019} and ant colony optimization
(ACO)~\cite{DNFW:07,NW:07,ST:2012}. 
Most of the results relevant to this work concern the binary PSO algorithm
and the \oneea algorithm.   For
the binary PSO, \cite{SW:08,SW:10} provide various runtime
results. For instance, they give general lower bound of $\Omega(n/\log n)$ for
every function with a unique global optimum and a bound of $\Theta(n \log n)$
on the function \onemax.  Note that the binary PSO studied in~\cite{SW:10} has
been designed for optimizing over $\{0,1\}^n$ and it is different from our
proposed \dpso, which can be applied to a much wider range of discrete
problems.  Sudhold and Witt show the following bound for the binary PSO.

\begin{theorem}\textbf{\emph{\cite[Thm.~3]{SW:10}}}
  Under certain as\-sump\-tions
  on the algorithm parameters, the expected optimization time of the binary PSO for optimizing $f : \{0,1\}^n
  \rightarrow \real$  is $O(mn\log n) + \sum_{i=1}^{m-1} 1/s_i$, where $m$ is the number of level sets of $f$ and $s_i$ is a lower bound on the probability to move from level $i$ to level $i-1$.
  \label{thm:sw10:general}
\end{theorem}

Essentially, this result reflects the fact that the binary PSO
converges to the attractor in expected time $O(n\log n)$ unless the attractor
has been updated meanwhile. This happens once for each fitness level. For
\onemax, this result yields an expected optimization time of $O(n^2 \log n)$. By a more careful
analysis of the binary PSO on \onemax, the following improved
bound is established:

\begin{theorem}\textbf{\emph{\cite[Thm.~5]{SW:10}}}\label{thm:sw10:onemax}
The expected
op\-ti\-miza\-tion time of the binary PSO with a single particle optimizing \onemax is $O(n \log n)$.
\end{theorem}

The \oneea
considered in~\cite{STW:04} is reminiscent of stochastic hill climbing: In each
iteration, a random solution is sampled and the current solution is replaced if
and only if the solution is better.  In order to escape local optima, the
distance between the current solution and the new one is determined according
to Poisson distributed random variables.
\cite{STW:04} provide bounds on the expected optimization time of a
\oneea sorting $n$ items. They consider various choices of objective functions (e.\,g., Hamming
distance, transposition distance, \ldots) as well as mutation operators (e.\,g.,
transpositions, reversing keys in a certain range, \ldots). A
general lower bound of $\Omega(n^2)$ is proved, which holds for all permutation
problems having objective functions with a unique
optimum~\cite[Thm.~1]{STW:04}. The most relevant runtime result for a
comparison with our \dpso algorithm is the following:

\begin{theorem}\textbf{\emph{\cite[Thm.~2/Thm.~4]{STW:04}}}
  The expected op\-ti\-miza\-tion
  time of the \oneea for sorting $n$ items is $\Theta(n^2
  \log n)$ if the objective function is the transposition distance to the
  sorted sequence and mutations are transpositions.
  \label{thm:easorting}
\end{theorem}

The upper bound can be obtained by the fitness level method and a lower bound of
$\Omega(k/n^2)$ on the probability of improvement when the current solution is at
transposition distance $k$ to the attractor. The lower bound follows from a
similar argument. In addition, for determining the lower bound \cite{STW:04} consider the
Hamming distance to evaluate the distance between the current position and the optimum
although the algorithm still uses the transposition distance to decide which
position is better. In the conference version \cite{MRSSW:17} we incorrectly
claimed that the proof does not apply to this setting. Thanks to an anonymous
reviewer we can correct this statement here.

In contrast to the \oneea algorithm, the binary PSO studied in~\cite{SW:10}
allows for non-improving solutions, but it converges to the attractor exactly
once per fitness level. After the convergence occurred, the binary PSO behaves
essentially like the \oneea.

Additionally, \cite{RSW:19} is based on a preliminary version of the present paper.
There the authors applied the \onepso to the single-source shortest path problem.
For this purpose they extend the Markov model presented here by allowing self loops.
They also used the bounds by integration which are proved in this work without repeating the proof.
The following upper and lower bounds on the expected optimization time are
given which are dependent on the algorithm parameter $c$ specifying the
probability of movement towards the attractor.
\begin{theorem}\textbf{\emph{\cite[Thm.~5/Thm.~7]{RSW:19}}}
  The expected optimization time $T(n)$, to solve the single-source shortest path problem with $n$ nodes is bounded by
  \begin{equation*}
    T(n)=\begin{cases}
      O(n^3)&\text{if }c\in(\frac{1}{2},1]\\
      O(n^{\nicefrac{7}{2}})&\text{if }c=\frac{1}{2}\\
      O(n^4\cdot \varphi(c)^n)&\text{if }c\in(0,\frac{1}{2})\\
    \end{cases}
  \end{equation*}
    and
  \begin{equation*}
    T(n)=\begin{cases}
      \Omega(n^2)&\text{if }c\in(\frac{1}{2},1]\\
      \Omega(n^{\nicefrac{5}{2}})&\text{if }c=\frac{1}{2}\\
      \Omega( (\varphi(c)-\varepsilon)^n)&\text{if }c\in(0,\frac{1}{2})\\
    \end{cases}
    \enspace,
  \end{equation*}
  where $\varphi(c)=e^{-(1-2c)/(1-c)}\cdot(\frac{1-c}{c})$ and any arbitrarily small $\varepsilon>0$.
\end{theorem}

\subsection{Organization of the Paper}

In Section~\ref{sec:dpso} we introduce the algorithm \dpso for solving
optimization problems over discrete domains.
In Section~\ref{sec:model} we provide a Markov model for the behavior of the
algorithm {\onepso} --- a restriction of \dpso to one particle --- between two
updates of the local attractor.
Section~\ref{sec:tools} contains a comprehensive analysis of this Markov model.
The results from this section are used in Section~\ref{sec:runtime} in
order to obtain the bounds on the expected optimization time for \onepso shown in
Table~\ref{tab:summary} as well as lower bounds for \dpso on pseudo-Boolean
functions. Section~\ref{sec:conclusion} contains some concluding remarks.

\section{Discrete PSO Algorithm}
\label{sec:dpso}

In this section we introduce the \dpso algorithm, a PSO
algorithm that optimizes functions over discrete domains.
A simplified version of the algorithm that uses just a single particle will be referred to as \onepso.
Note that \onepso is different from the 1-PSO studied in \cite{SW:10}, which
is tailored to optimization over bitstrings.
The \dpso and \onepso algorithm sample items from a finite
set $X$ in order to determine
some $x^{*} \in X$ that minimizes a given \emph{objective function} $f:X
\longrightarrow \real$. In order to have a discrete PSO that remains true to
the principles of the original PSO for optimization in the
domain $\real^n$ from~\cite{eb_ken_1995,ken_eb_1995}, we need some additional
structure on $X$: For each $x \in X$ we have a set of \emph{neighbors} $\neighborhood_X(x)$.
If the set $X$ is clear from the context we may drop the subscript. The
neighborhood structure induces a solution graph with nodes $X$ and arcs $\{xy
\mid x, y \in X, y \in \neighborhood(x)\}$. The \emph{distance} $\dist(x, y)$ of solutions
$x, y \in X$ is the length of a shortest (directed) $xy$-path in this graph.
We assume that the solution graph is strongly connected, so the PSO cannot get
``trapped'' in a particular strongly connected component. The search spaces of
our reference problems in combination with the used neighborhood relationship
satisfy this assumption.

\begin{algorithm}[ht]
  \caption{\dpso}
  \label{alg:dpso}
  \SetStartEndCondition{ }{}{}%
  \SetKwFunction{Range}{range}%%
  \SetKw{KwTo}{to}%
  \SetKw{KwAnd}{and}%
  \SetKwFor{For}{for}{ do}{}%
  \SetKwIF{If}{ElseIf}{Else}{if}{ then}{else if}{else}{}%
  \SetKwFor{While}{while}{ do}{fintq}%
  \AlgoDontDisplayBlockMarkers\SetAlgoNoLine\SetAlgoNoEnd%\SetAlgoVlined%
  \DontPrintSemicolon
  \SetKwInOut{Input}{input}
  \SetKwInOut{InOut}{in/out}
  \SetKwInOut{Output}{output}
  \Input{Function $f:X \rightarrow \real$, $P\in\naturalNumbers$,\\ $c_{loc} \in [0,1]$, $c_{glob}\in[0,1-c_{loc}]$}
  \BlankLine\;

\For(\tcc*[f]{initialization})
{$i=1$ \KwTo $P$}
{
    pick position $x_i \in X$ \uar\;
    $l_i \longleftarrow x_i$\tcc*{local attractor}
 
}

$g\longleftarrow {\mathrm{argmin}}_{x_i}\lbrace f(x_i)\rbrace$\tcc*{\!\!\!global attractor\!\!\!}
\While(\tcc*[f]{PSO iterations})
{\texttt{True}}
{
    \For
    {$i=1$ \KwTo $P$}
    {
        pick $q \in [0,1]$ \uar\;
        \uIf
        {$x_i \neq l_i$ \KwAnd $l_i \neq g$ \KwAnd $q \in[0,c_{loc}]$}
        {
          \tcc{towards local attractor (if not equal to global attractor)}
            \nl \label{alg:sample:local}$x' \in \{ y \in \neighborhood(x_i) \mid \dist(y,l_i) < \dist(x_i, l_i)\}$ \!\uar\!\!\! \;
        }
        \uElseIf
        {$x_i \neq g$ \KwAnd $q \in]1-c_{glob},1]$}
        {
            \tcc{towards global attractor}
            \nl \label{alg:sample:global}$x' \in \{ y \in \neighborhood(x_i) \mid \dist(y,g) < \dist(x_i, g)\}$\,\uar\;
        }
        \uElse(\tcc*[h]{random direction})
        {
            \nl \label{alg:sample:random}$x' \in \neighborhood(x_i)$ \uar\;
        }
        $x_i \longleftarrow x'$\tcc*{update position}
        \tcc{update attractors}
        \If
        {$f(x_i) < f(l_i)$}
        {
            \nl \label{alg:improvement:local} $l_i \longleftarrow x_i$\;
        }
            \If
            {$f(x_i) < f(g)$}
            {
                \nl \label{alg:improvement:global} $g \longleftarrow x_i$\;
            }
    }
}
\end{algorithm}

The \dpso algorithm performs the steps shown in Algorithm~\ref{alg:dpso}.  The
initial positions of the particles are chosen uniformly at random (\uar) from
the search space~$X$. The parameter $c_{loc}$ determines the importance of the local attractor
$l_i$ of each particle, which is the best solution a single particle has found so
far, and the parameter $c_{glob}$ determines the importance of the global
attractor $g$, which is the best solution all particles have found so far.  In each
iteration each particle moves towards the local attractor with probability
$c_{loc}$, moves towards the global attractor with probability $c_{glob}$ and
otherwise move to a random neighbor. If $l_i$ equals $g$ then the particles
still move only with probability $c_{glob}$ to $g$.
Note that the attractors $l_i$ and $g$
are updated in lines~\ref{alg:improvement:local}
and~\ref{alg:improvement:global} whenever a strictly better solution has been found.

Alternatively, one could choose to update local and global attractors whenever
the new position is at least as good as the position of the attractor (use $\leq$
instead of $<$ in lines~\ref{alg:improvement:local}
and~\ref{alg:improvement:global}). The theorems presented here can also be
carried over to this modified setting. Admittedly, for functions with plateaus
this version potentially performs better, since a plateau is traversed easily
by the modified algorithm. However, for the problems considered in this work
there are no plateaus. Additionally, the probability to improve the objective function in
the situation where position and attractor differ but have equal objective
function value is higher than in the situation where the position is placed at
the attractor. 

At first glance, Algorithm~\ref{alg:dpso} may not seem like an implementation of the PSO ideas, since we are choosing only a
single attractor on each move or even make a random move whereas the classical
PSO uses local and global attractor at each move, but looking at several
consecutive iterations we retain the tendency of movement towards all the
attractors.  We consider the PSO to be an infinite process, so we do not give a
termination criterion. 
We assume that sampling of $x'$ in lines~\ref{alg:sample:local},
\ref{alg:sample:global} and~\ref{alg:sample:random} can be performed
efficiently. This is the case for the neighborhood structures we consider.

The algorithm \onepso is simply given by Algorithm~\ref{alg:dpso} with a single particle, i.e., we have $P = 1$.
Note that there is only a
single attractor in this case. Hence, \onepso has just a single parameter $c=c_{glob}$ that
determines the probability of moving towards the (global) attractor $g$. In all other
aspects it behaves like the \dpso algorithm.

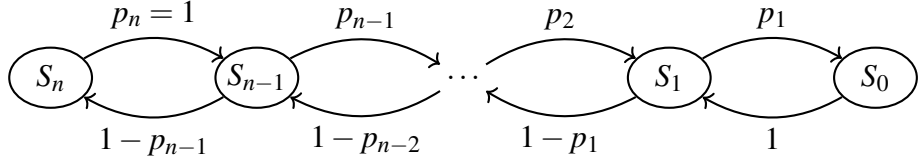
\begin{figure*}[t]
  \centering
  \begin{tikzpicture}[node distance=6.5em,vertex/.style={shape=ellipse, minimum height=0.8cm, minimum width=1.1cm}]
	\node[draw, thick, vertex] (s4) {};
	\node[draw, thick, vertex,right of=s4] (s3) {};
	\node[      thick, right of=s3] (pause) {$\ldots$};
	\node[draw, thick, vertex,right of=pause] (s1) {};
	\node[draw, thick, vertex,right of=s1] (s0) {};

	\node[vertex] (s4t) at (s4) {$S_{n}$};
	\node[vertex] (s3t) at (s3) {$S_{n-1}$};
	\node[vertex] (s1t) at (s1) {$S_1$};
	\node[vertex] (s0t) at (s0) {$S_0$};

	\path[thick,->]
		(s4)    edge [bend left=30] node[label={[label distance=-5pt]above:{$p_{n}=1$}}] {} (s3)
		(s3)    edge [bend left=30] node[label={[label distance=-5pt]above:{$p_{n-1}$}}] {} (pause)
		(pause) edge [bend left=30] node[label={[label distance=-5pt]above:{$p_{2}$}}] {} (s1)
		(s1)    edge [bend left=30] node[label={[label distance=-5pt]above:{$p_{1}$}}] {} (s0)
        ;

	\path[thick,->]
	  (s0)    edge [bend left=30] node[label={[label distance=-5pt]below:{$1$}}] {} (s1)
	  (s1)    edge [bend left=30] node[label={[label distance=-5pt]below:{$1-p_{1}$}}] {} (pause)
	  (pause) edge [bend left=30] node[label={[label distance=-5pt]below:{$1-p_{n-2}$}}] {} (s3)
	  (s3)    edge [bend left=30] node[label={[label distance=-5pt]below:{$1-p_{n-1}$}}] {} (s4);
  \end{tikzpicture}
  \caption{State diagram of the Markov model}
  \label{fig:markov}
\end{figure*}

\section{Markov Model of \onepso}
\label{sec:model}

We present a simple Markov model that captures the behavior of the \onepso
algorithm between two consecutive updates of the attractor. This model has been
presented already in \cite{MRSSW:17} but we repeat it here to present a self
contained overview on the presented approach. As an extension to
\cite{MRSSW:17} we also present how the variance can be computed. This is
essential for experiments and if \onepso and \dpso are actually used.
Using this model we
can infer upper and lower bounds on the expected optimization time of
the \onepso algorithm on suitable discrete functions.
For our analysis, we assume that the objective function $f: X \rightarrow
\real$ has the property that every local optimum is a global one. That is, the
function is either constant or any non-optimum solution $x$ has a neighbor $y
\in \neighborhood(x)$ such that $f(x) > f(y)$. Functions with that property are
called \emph{unimodal} functions. Although this restriction
certainly narrows down the class of objective functions to which our analysis
applies, the class seems still appreciably large, e.\,g., it properly contains  
the class of functions with a unique global optimum and no further local optima. 

Assume that the attractor $g \in X$ is fixed and $g$ is not a minimizer of $f$.
Under which conditions can a new ``best'' solution be found? Certainly, if the
current position $x$ is equal to $g$, then, by the described structure of $f$ we get an
improvement with positive probability. If $x \neq g$ then the attractor may
still be improved. However, for the purpose of upper bounding the expected
optimization time of the \onepso we dismiss the possibility that the attractor
is improved if $x \neq g$. As a result, we obtain a reasonably simple
Markov model of the \onepso behavior. Quite surprisingly, using the same Markov model,
we are also able to get good lower bounds on the expected optimization time of
the \onepso (see Section~\ref{sec:runtime:lb} for the details).

Recall that we think of the search space in terms of a strongly connected graph.
Let $n$ be the diameter of the search space $X$, i.e., the maximum distance of any two points in $X$.  We partition the search space
according to the distance to the attractor $g\in X$. That is,
for $0 \leq i \leq n$, let $X_i  = \{ x \in X \mid \dist(x, g) = i\}$. Note
that this partition does not depend on the objective function. If the search
space is not symmetric as in our case it could also be possible that some $X_i$
are empty, because the maximal distance to a specific solution in the search
space could be less than the diameter. The model consists of $n+1$ states
$S_0,S_1,\ldots,S_n$. Being in state $S_i$ indicates that the current solution
$x$ of the \onepso is in $X_i$.

For each $x\in X_i$ we denote by $p_x$ the transition probability from $x$ to an element in $X_{i-1}$.
The probabilities $p_x$ in turn depend on the parameter $c$, which
is the probability that the \onepso explicitly moves towards the attractor. 
If the current position $x$ is in $X_{i}$ and
the algorithm moves towards the attractor, then the new position is in $X_{i-1}$.
On the other hand, if the PSO updates $x$ to any neighbor chosen
\uar from $\neighborhood(x)$, then the new position is in $X_{i-1}$ with probability
$|\neighborhood(x) \cap X_{i-1}|/|\neighborhood(x)|$.
So we obtain the transition probability
\[
  p_x = c + (1-c) \cdot\frac{|X_{i-1} \cap \neighborhood(x)|}{|\neighborhood(x)|} \enspace.
\]
\begin{remark}
  \label{rem:noInternalTransitions}
  In this work we assume that the probability that we move from a position $x\in X_i$ to an element in $X_i$ is zero, i.\,e., if we move from a position $x$ to a neighboring position $x'\in\neighborhood(x)$ then the distance to the attractor does always change ($\dist(x,g)\neq\dist(x',g)$).
\end{remark}
This assumption holds for both problems we investigate in
Section~\ref{sec:runtime}. Nevertheless, an extensions allowing transitions
inside a level $X_i$ is possible and has been considered already in the
article~\cite{RSW:19} which is based on an arXiv preprint of this article.

Using the assumption in Remark~\ref{rem:noInternalTransitions} and the fact that
$X_i$ is defined by distance to a fixed position $g$ the probability of moving
from $x$ to an element in $X_j$, $j \notin \{i-1, i+1\}$, is zero.
Consequently, the probability of moving from $x$ to an element in $X_{i+1}$ is then $1-p_x$.
Furthermore, if the \onepso is at position $x \in X_n$ then any move brings us
closer to the reference solution; so $p_x = 1$ in these cases.

Please note that $p_{x}$ and $p_{x'}$ can differ even if $x,x'$ are contained
in the same set $X_i$.  Therefore we do not necessarily obtain a Markov model
if we use the states $S_i$ and the transition probabilities $p_i=p_x$ for some
$x\in X_i$ as this value is not necessarily equal for all $x'\in X_i$.

Nevertheless, we can analyze Markov chains using bounds on the transition
probabilities.  To be more precise, we can use $p_i:=\min_{x\in X_i}p_x$ as
lower bound and $p_i':=\max_{x\in X_i}p_x$ as upper bound on the transition
probabilities in the direction to the attractor to obtain an upper bound and
lower bound on the expected number of iterations until the distance to the
attractor is decreased respectively.
Figure~\ref{fig:markov} shows the state diagram of this model.

\begin{definition}
  \label{defi:model}
  By
  $\model\left((p_i)_{1 \leq i \leq n}\right)$ we denote an instance of the
  Markov model with states $S_0,S_1,\ldots,S_n$ and data $p_i\in\R_{\geq 0}$, $1 \leq i \leq n$.
  Suppose we are in state $S_i$. Then we move to state $S_{i-1}$ with probability
  \begin{align*}
    1 &\text{ if }   i = n\\
    \min \{ 1, p_i \} &\text{ if }   1 \leq i < n
  \end{align*}
  and otherwise we move to state $S_{i+1}$.
\end{definition}
Please note that the data does not need to be in $[0,1]$ but for the ease of
presentation we will refer to them as probabilities. Furthermore, supposed we
are in state $S_{n}$ even if $p_n\neq 1$ the probability of moving to $S_{n-1}$
is $1$. This notation allows us to succinctly specify the currently used
Markov model, e.\,g., the model which is used to obtain upper bounds is
described by $\model\left( (\min_{x\in X_i}p_x)_{1\leq i\leq n} \right)$.

Our goal is to determine the expected number of steps needed to hit a solution
which is better than the attractor after starting in $S_0$. Let $p_g$ be the
probability to improve the attractor if we are currently in state $S_0$, hence
at the attractor. Then the probability $p_g$ depends on $f$ and the choice of
$g$. We have that $p_g$ is positive since $f$ is unimodal.  In order to reach a
better solution from $S_0$ we need in expectation $1/p_g$ tries. If we are
unsuccessful in some try, then the \onepso moves to $S_1$. For upper bounds we
can ignore the chance to improve the attractor through other states. Thus we
need to determine the expected number of steps it takes until we can perform
the next try, that is, the expected first hitting time for the state $S_0$,
starting in $S_1$. The expected number $h_i$ of steps needed to move from $S_i$
to $S_0$ is given by the following recurrence:
\begin{equation}
  \begin{aligned}
	h_{n} &= 1 + h_{n-1}\enspace ,\enspace h_0=0 \\
	h_{i} &= 1 + p_i \cdot h_{i-1} + (1 - p_i)\cdot h_{i+1}\enspace, & 1 \leq i < n\enspace.
  \end{aligned}
  \label{eq:recmarkov}
\end{equation}

\section{Analysis of the Markov Model}
\label{sec:tools}

In this section we prove upper and lower bounds on the expected return time to
the state $S_0$ of the Markov model from Section~\ref{sec:model}. These bounds
are of key importance for our runtime analysis in Section~\ref{sec:runtime}.
For the lower bounds we also introduce a notion of indistinguishability of
certain states of a Markov model.

In our analysis the probabilities $p_i$ are generally not identical. If
we assume that $p_1 = p_2 = \ldots = p_n = p$, then we obtain a non-homogeneous
recurrence of order two with constant coefficients. In this case, standard
methods can be used to determine the expected time needed to move to the
attractor state $S_0$ as a function of $n$~\cite[Ch.~7]{GKP:94}. Note also that
for $p=1/k$ this is exactly the
recurrence that occurs in the analysis of a randomized algorithm for
$k$-\textsc{SAT}~\cite{P:91,S:99} and~\cite[pp.~160f.]{MU:05}.
If $p_i$ has not necessarily identical values which are dependent on $i$, then the recurrence can in some
cases be solved, see e.\,g., \cite[Ch.~7]{GKP:94} and~\cite{Petkovsek:94}. Here,
due to the structure of the recurrence, we can use a more pedestrian approach,
which is outlined in the next section.

\subsection{Reformulation of the Recurrence}
\label{sec:reformulation}

We first present a useful reformulation of Recurrence~\eqref{eq:recmarkov}. 
From this reformulation we will derive closed-form expressions and asymptotic
properties of the return time to the attractor of the transition probabilities.

Let $W_{i}$ be the number of steps needed to
move from state $S_i$ to state $S_{i-1}$ and let $H_i:=\E[W_i]$ be its
expectation.  Then $H_i$ can be determined from
$H_{i+1}$ as follows: In expectation, we need $1/p_i$ trials to get from $S_i$
to $S_{i-1}$, and each trial, except for the successful one, requires
$1+H_{i+1}$ steps. The successful trial requires only a single step, so
$H_{i}$ is captured by the following recurrence:
\begin{align}
  H_i &=  \frac{1}{p_i} \left( 1 + H_{i+1} \right) - H_{i+1}\enspace =\frac{1}{p_i}+\frac{1-p_i}{p_i}\cdot H_{i+1} , & 1 \leq i < n \label{eq:recreform}\\
  H_n &= 1\enspace. \label{eq:recinitial}
\end{align}
Another interpretation is the following: 
$W_i$ is equal to one
with probability $p_i$, the direct step to $S_{i-1}$, and with probability
$1-p_i$ it takes the current step which leads to state $S_{i+1}$, then
$W_{i+1}$ steps to go from $S_{i+1}$ back to
$S_i$ and then again $W_i$ 
steps. For the expected value of $W_i$ this interpretation leads to the formula
$$
H_i=p_i + (1-p_i)\cdot(1+H_{i+1}+H_i)\enspace,
$$
which is equivalent to Equation~\ref{eq:recreform} after solving for $H_i$. 
Please note that the probabilities $p_i$ are mostly determined by some
function depending on $n$ and $i$.
Unfolding the recurrence specified in Equation~\eqref{eq:recreform} $k$ times,
$1 \leq k \leq n$, followed by some rearrangement of the terms yields 
\begin{equation}
  H_1=\sum_{i=1}^{k-1}\left(  \frac{1}{p_i}\cdot\prod_{j=1}^{i-1}\frac{1-p_j}{p_j}\right) + H_k\cdot\prod_{j=1}^{k-1}\frac{1-p_j}{p_j}\enspace.\label{eq:partialunfold}
\end{equation}
Thus, for $k=n$ we obtain the following expression for $H_1$:
\begin{equation}
  H_1 = \sum_{i=1}^n \left( \frac{1}{p_i} \cdot \prod_{j=1}^{i-1} \frac{1-p_j}{p_j}\right) - \prod_{j=1}^{n}\frac{1-p_j}{p_j}\enspace,
  \label{eq:recclosed}
\end{equation}
where the second term is a correction term which is required whenever $p_n <
1$ (see Definition \ref{defi:model}) in order to satisfy the initial condition given in
Equation~\eqref{eq:recinitial}.  Equation~\eqref{eq:recclosed} has also been
mentioned in \cite[Lemma 3]{DJW:01} in the context of the analysis of randomized local search or in \cite[Theorem 3]{KK:18}.
$H_k$ can be obtained analogously, which leads to
\begin{equation}
  H_k = \sum_{i=k}^n \left( \frac{1}{p_i} \cdot \prod_{j=k}^{i-1} \frac{1-p_j}{p_j}\right) - \prod_{j=k}^{n}\frac{1-p_j}{p_j}\enspace.
  \label{eq:recclosedk}
\end{equation}

\subsection{Identical Transition Probabilities} \label{section:constant_probabilities}
\label{subsec:const}

If the probabilities $p_i = p$ for some constant $p \in [0,1]$ and $1 \leq i <
n$, then Recurrences~\eqref{eq:recreform} become linear recurrence equations
with constant coefficients. Standard methods can be used to determine
closed-form expressions for $h_i$ and $H_i$. However, we are mainly interested
in $H_1$ and are able to determine closed-form expressions directly from
Equation~\eqref{eq:recclosed}. 

\begin{theorem}
  \label{thm:returntime}
  Let $0 < p < 1$. Then the expected return time $H_1$ to $S_0$ is
  \begin{equation}
	H_1 = h_1 = 
	\begin{cases}
	  \dfrac{1-2p \left( \dfrac{1-p}{p} \right)^n }{2p-1}	&	\text{if $p \neq \frac{1}{2}$}\\
	  2n-1															 &   \text{if $p = \frac{1}{2}$}\enspace.
	\end{cases}
	\label{eq:constreturn}
  \end{equation}
\end{theorem}
\begin{proof}
  By setting $p_i = p$ in Equation~\eqref{eq:recclosed} and performing some
  rearrangements the theorem is proved.
\end{proof}

It is easily verified that this expression for $h_1$ satisfies
Equation~\eqref{eq:recmarkov}.  So, with $p_i = p$ we have that the time it takes to
return to the attractor is bounded from above by a constant, a linear function,
or an exponential exponential function in $n$ if $p > 1/2$, $p = 1/2$, or $p <
1/2$, respectively.

For the case $p=1/2$ one can obtain this result also from the Gambler's Ruin
Model by mirroring the state $0$ at $n$ which would result in termination if
values $0$ or $2n$ appear. Starting at value $1$ results in a first hitting
time of $2n-1$ in this Gambler's Ruin Model as stated in
Theorem~\ref{thm:returntime}.

\subsection{Non-identical Transition Probabilities}
\label{subsec:nonconst}

Motivated by the runtime analysis of \onepso applied to optimization problems
such as sorting and \onemax, we are particularly interested in the expected
time it takes to improve the attractor if the probabilities $p_i$ are
\emph{slightly} greater than $1/2$. By \emph{slightly} we mean $p_i = 1/2 +
i/(2A(n))$ which appears in the analysis of \onepso optimizing \onemax and the
sorting problem, or $p_i=1/2+A(i)/(2A(n))$ which appears in the
analysis of \onepso optimizing the sorting problem,
where  $A:\nat \rightarrow \nat$ is some non-decreasing function of $n$ such
that $\lim_{n \rightarrow \infty} A(n) = \infty$.
Recall from Definition~\ref{defi:model} that if $p_i>1$ then we move from state
$S_i$ to state $S_{i-1}$ with probability $1$.  Clearly, in this setting we
cannot hope for a recurrence with constant coefficients. Our goal in this
section is to obtain the asymptotics of $H_1$ as $n \rightarrow \infty$ for
$A(n) = n$ and $A(n) = \binom{n}{2}$. We show that for $p_i=1/2+i/(2\cdot
A(n))$ and $A(n) = n$ the return time to the attractor is $\Theta(\sqrt{n})$,
while for $A(n) = \binom{n}{2}$ the return time is $\Theta(n)$.
\begin{lemma}
  Let $M=\model((1/2+i/(2n))_{1\leq i\leq n})$. Then 
  \[
	H_1  = \frac{4^n}{\binom{2n}{n}} - 1
        \sim\sqrt{\pi n}=\Theta(\sqrt{n})\enspace.
  \]
  \label{lemma:sqrt_return}
\end{lemma}
\begin{proof}
  We have $p_n = 1$ so the correction term in Equation~\eqref{eq:recclosed} is zero.
  We rearrange the remaining terms of Equation~\eqref{eq:recclosed} and find that
\begin{align*}
  H_1	&= \sum_{i=1}^n \frac{2n}{n + i} \prod_{j=1}^{i-1} \frac{n-j}{n+j}
  = 2\,  \sum_{i=1}^n \frac{n!\,n!}{(n-i)!\cdot(n+i)!} 
  \overset{i'=n-i}{=} \frac{2}{\binom{2n}{n}} \sum_{i'=0}^{n-1} \binom{2n}{i'}
   = \frac{4^n}{\binom{2n}{n}} - 1\enspace.
\end{align*}
Applying the well-known relation
$$
\frac {4^n}{\sqrt{\pi n}} \left(1-\frac 1 {4n}\right) 
<\binom{2n}{n}
< \frac {4^n}{\sqrt{\pi n}}
$$
(for an elegant derivation, see~\cite{H:15})
finishes the proof.
\end{proof}
This lemma can be generalized for linearly growing probabilities.
\begin{theorem}
  Let $M = \model\left( (p_i)_{1 \leq i \leq n}\right)$, where $p_i = 1/2 + i
  /(2A(n))$. 
  \\Then $H_1 = \Theta(\min(\sqrt{A(n)},n))$ with respect to $M$.
  \label{thm:h1_lin_theta}
\end{theorem}
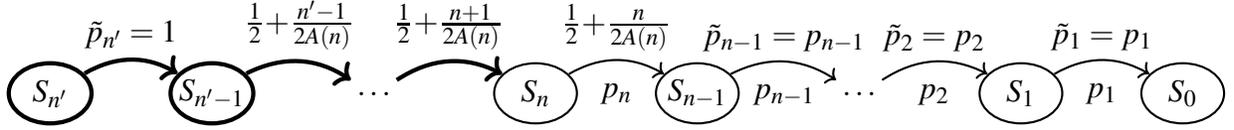
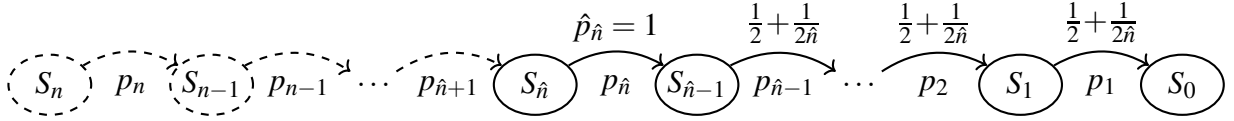
\begin{figure*}[tb]
\begin{center}
\subfloat[][$M$ (below arrows) and its extended version $\tilde M$ (above arrows) with $n'>n$ states.]
{
  \centering
  \begin{tikzpicture}[node distance=5.4em,vertex/.style={shape=ellipse, minimum height=0.8cm, minimum width=0.9cm}]
	\node[draw, ultra thick, vertex] (s7){};
	\node[draw, ultra thick, vertex,right of=s7    ] (s6) {};
	\node[            thick,        right of=s6    ] (pause2) {$\ldots$};
	\node[draw,       thick, vertex,right of=pause2] (s4) {};
	\node[draw,       thick, vertex,right of=s4    ] (s3) {};
	\node[            thick,        right of=s3    ] (pause) {$\ldots$};
	\node[draw,       thick, vertex,right of=pause ] (s1) {};
	\node[draw,       thick, vertex,right of=s1    ] (s0) {};
    \node[vertex] (s7t) at (s7) {$S_{n'}$};
	\node[vertex] (s6t) at (s6) {$S_{n'-1}$};
	\node[vertex] (s4t) at (s4) {$S_{n}$};
	\node[vertex] (s3t) at (s3) {$S_{n-1}$};
	\node[vertex] (s1t) at (s1) {$S_1$};
	\node[vertex] (s0t) at (s0) {$S_0$};

	\path[thick,->]
		(s7)     edge [ultra thick, bend left=30] node[label={[label distance=-3pt] above:{$\tilde p_{n'}=1$}}] {} (s6)
		(s6)     edge [ultra thick, bend left=30] node[label={[label distance=-3pt] above:{$\frac{1}{2}\!+\!\frac{n'-1}{2A(n)}$}}] {} (pause2)
		(pause2) edge [ultra thick, bend left=30] node[label={[label distance=-3pt] above:{$\frac{1}{2}\!+\!\frac{n+1}{2A(n)}$}}] {} (s4)
		(s4)     edge [             bend left=30] node[label={[label distance=-3pt] above:{$\frac{1}{2}\!+\!\frac{n}{2A(n)}$}}, label={below:{$p_{n}$}}] {} (s3)
		(s3)     edge [             bend left=30] node[label={[label distance=-3pt] above:{$\tilde p_{n-1}=p_{n-1}$}}, label={below:{$p_{n-1}$}}] {} (pause)
		(pause)  edge [             bend left=30] node[label={[label distance=-3pt] above:{$\tilde p_{2}=p_{2}$}}, label={below:{$p_2$}}] {} (s1)
		(s1)     edge [             bend left=30] node[label={[label distance=-3pt] above:{$\tilde p_{1}=p_{1}$}}, label={below:{$p_1$}}] {} (s0)
        ;

  \end{tikzpicture}
  \label{fig:markovExtendedMTilde}
}\\
\subfloat[][$M$ (below arrows) and its reduced version $\hat M$ (above arrows) with $\hat n<n$ states and larger probabilities in direction to $S_0$.]
{
  \centering
  \begin{tikzpicture}[node distance=5.4em,vertex/.style={shape=ellipse, minimum height=0.8cm, minimum width=0.9cm}]
	\node[draw, dashed, thick, vertex] (s7) {};
	\node[draw, dashed, thick, vertex,right of=s7] (s6) {};
	\node[              thick,        right of=s6] (pause2) {$\ldots$};
	\node[draw,         thick, vertex,right of=pause2] (s4) {};
	\node[draw,         thick, vertex,right of=s4] (s3) {};
	\node[              thick,        right of=s3] (pause) {$\ldots$};
	\node[draw,         thick, vertex,right of=pause] (s1) {};
	\node[draw,         thick, vertex,right of=s1] (s0)    {};
	\node[vertex] (s7t) at (s7) {$S_{n}$};
	\node[vertex] (s6t) at (s6) {$S_{n-1}$};
	\node[vertex] (s4t) at (s4) {$S_{\hat n}$};
	\node[vertex] (s3t) at (s3) {$S_{\hat n-1}$};
	\node[vertex] (s1t) at (s1) {$S_1$};
	\node[vertex] (s0t) at (s0) {$S_0$};

	\path[thick,->]
		(s7)     edge [dashed, bend left=30] node[label={below:{$p_{n}$}}] {} (s6)
		(s6)     edge [dashed, bend left=30] node[label={below:{$p_{n-1}$}}] {} (pause2)
		(pause2) edge [dashed, bend left=30] node[label={below:{$p_{\hat n+1}$}}] {} (s4)
		(s4)     edge [        bend left=30] node[label={[label distance=-3pt] above:{$\hat p_{\hat n}=1$}}, label={below:{$p_{\hat n}$}}] {} (s3)
		(s3)     edge [        bend left=30] node[label={[label distance=-3pt] above:{$\frac{1}{2}\!+\!\frac{1}{2\hat n}$}}, label={below:{$p_{\hat n-1}$}}] {} (pause)
		(pause)  edge [        bend left=30] node[label={[label distance=-3pt] above:{$\frac{1}{2}\!+\!\frac{1}{2\hat n}$}}, label={below:{$p_2$}}] {} (s1)
		(s1)     edge [        bend left=30] node[label={[label distance=-3pt] above:{$\frac{1}{2}\!+\!\frac{1}{2\hat n}$}}, label={below:{$p_1$}}] {} (s0)
        ;

  \end{tikzpicture}
  \label{fig:markovReducedMHat}
}
  \end{center}
 
  \caption{State diagram of the Markov model $M$ and its modified versions $\tilde M$ and $\hat M$ used in the proof of Theorem~\ref{thm:h1_lin_theta}.}
  \label{fig:markovModifiedModels}
\end{figure*}
\begin{proof}
  First, we consider the case $A(n) \leq n^2$. Let $n' =  A(n)$, which is the smallest
  number such that $p_{n'} = 1/2 + n'/(2A(n)) \geq 1$. First, assume that $n'
  \leq n$ and consider the ``truncated'' model $M' = \model\left( (p_i)_{1 \leq
  i \leq n'} \right)$.  Please note that there is actually no difference
  between $M$ and $M'$ because the removed states are never visited as
  $p_{n'}$, the probability to move from $S_{n'}$ to $S_{n'-1}$, is already one
  and $S_{n'+1}$ is never visited.  Let $H_1'$ be the expected time to
  reach state $S_0$ starting at state $S_1$ with respect to $M'$. By
  Lemma~\ref{lemma:sqrt_return} we have $H_1'=\THETA(\sqrt{A(n)})$, which is by
  the construction of $M'$ equal to $H_1$. On the other hand, assume that
  $n' > n$ and consider the ``extended'' model $\tilde M = \model\left(
  (p_i)_{1 \leq i \leq n'} \right)$. $M$ and $\tilde M$ are visualized in
  Figure \ref{fig:markovExtendedMTilde} with omitted probability of worsening.
  Let $\tilde H_1$ be the expected time to reach state $S_0$ starting at state
  $S_1$ with respect to $\tilde M$. By Lemma~\ref{lemma:sqrt_return} we have
  $\tilde H_1 = \THETA(\sqrt{A(n)})$ and since $\tilde H_1 \geq H_1$ we obtain
  $H_1 = O(\sqrt{A(n)})$.

  To obtain a lower bound on $H_1$ for the case $n' > n$ we consider the model
  $\hat M = \model\left( (\hat p_i)_{1\leq i\leq \hat n} \right)$, where $\hat
  n = \min(n,\lfloor\sqrt{A(n)}\rfloor)$ and $\hat p_i = 1/2 + 1/(2\hat n)$ for
  $1 \leq i < \hat n$, and $\hat p_{\hat n} = 1$. For $1\leq i < \hat n$ we
  have that $\hat p_i \geq p_i$ because $1/\hat n=\hat n/\hat n^2\geq \hat
  n/(A(n))$. A schematic representation of $M$ and $\hat M$ can be found in
  Figure \ref{fig:markovReducedMHat}. Let $\hat H_1$ denote the expected time
  to reach state $S_0$ starting at state $S_1$ in $\hat M$. Since $\hat p_i
  \geq p_i$ for $1 \leq i \leq \hat n$, $\hat H_1$ is a lower bound on $H_1$.
  This fact is obvious as in this Markov model one can move only to neighboring
  states. Increasing the probability of moving towards the final state
  consequently decreases the probability of movement in the opposite direction
  and therefore the hitting time to reach the final state is in expectation
  reduced if the probabilities of moving towards the final state are increased.
  Since $p := \hat p_i$ is constant for $1 \leq i < \hat n$ we get from
  Theorem~\ref{thm:returntime} that
  \[
	\hat H_1 = \frac{1-2p\left(\dfrac{1-p}{p}\right)^{\hat n}}{2p-1}\enspace .
  \]
  Substituting $p = 1/2 + 1/{(2 \hat n)}$ gives
  \begin{align*}
  \hat H_1 &= {\hat n}-({\hat n}+1)\cdot\left(\frac{{\hat n}-1}{{\hat n}+1}\right)^{\hat n}
  = {\hat n}-\frac{({\hat n}+1)^2}{{\hat n}-1}\cdot\left(1-\frac{2}{{\hat n}+1}\right)^{{\hat n}+1}
  \\&\geq{\hat n}-\frac{({\hat n}+1)^2}{{\hat n}-1}\cdot {\rm e}^{-2}
  =\Omega({\hat n})\enspace.
  \end{align*}
  Therefore $H_1=\Omega({\hat n})=\Omega(\min(n,\sqrt{A(n)}))$.

  It remains to show that the statement holds if $A(n) > n^2$. In this
  case, $H_1 = O(n)$ is obtained by setting $p_i = 1/2$, which is a lower bound
  on the probabilities of moving towards $S_0$, for $1 \leq i < n$ and
  invoking Theorem~\ref{thm:returntime}.  On the other hand, setting
  $\breve{A}(n) = n^2$ and using $\breve M=\model((1/2+i/(2\breve{A}(n)))_{1 \leq
  i \leq n})$ gives a lower bound on $H_1$, because $1/2+i/(2\breve{A}(n))$ is
  an upper bound on $1/2 + i/(2A(n))$. As discussed above, for
  the case ${A}(n) \leq n^2$, the expected time to reach state $S_0$ starting at
  state $S_1$ in $\breve M$ is $\OMEGA(\min(n,\sqrt{\breve{A}(n)}))$.  Therefore,
  $H_1 = \OMEGA(\min(n,\sqrt{\breve{A}(n)}))=\OMEGA(n)$, which completes the proof.
\end{proof}

For our application, the sorting problem, the following special case of
Theorem~\ref{thm:h1_lin_theta} will be of interest:

\begin{corollary}
Let $M = \model((1/2(1+i/\binom{n}{2}))_{1 \leq i \leq n})$, then $H_1=\Theta(n)$.
  \label{cor:returntime}
\end{corollary}

We will now consider a slightly different class of instances of the
Markov model in order to obtain a lower bound on the \onepso runtime for
sorting in Section~\ref{sec:runtime}. For this purpose we consider transition
probabilities $p_i$ that increase in the same order as the divisor $A(n)$,
which suits our fitness level analysis of the sorting problem. 
This class of models is
relevant for the analysis of the \emph{best case} behavior of the \onepso
algorithm for sorting $n$ items (see Theorem~\ref{thm:probabilities}).
Although we will only make use of a lower bound on $H_1$ in this setting later
on, we give the following $\THETA$-bounds:

\begin{theorem}
  Let $M = \model\left((\frac{1}{2}\left(1+A(i)/A(n)\right))_{1\leq i\leq n}\right)$ where
  $A:\naturalNumbers\rightarrow \real^+$ is a non-decreasing function and $A(n)=\THETA(n^d)$ for some $d> 0$.
  Then $H_1 =
  \THETA(n^{d/(d+1)})$ with respect to $M$.
  \label{thm:h1_quad_theta}
\end{theorem}

This theorem is a significant extension to \cite[Thm. 5]{MRSSW:17} which covers only the special case $p_i=1/2\cdot(1+\binom{i+1}{2}/\binom{n}{2})$. 
\begin{proof}
  Consider the expression for $H_1$ given in Equation~\eqref{eq:partialunfold}.
  Since $p_i > 1/2$ for $1 \leq i \leq n$ the products are at most $1$ and
  $1/p_i$ is at most $2$. Therefore for any $k\in\lbrace 1,\ldots,n\rbrace: 
  H_1\leq 2k+H_{k+1}$.  As $H_{k+1}$ is the expected number of steps to move
  from state $S_{k+1}$ to $S_k$, the states $S_0$ to $S_{k-1}$ are irrelevant
  for the calculation of $H_{k+1}$ since they are never visited in between.
  Therefore also probabilities $p_1$ to $p_{k}$ do not matter.  We truncate the
  model to states $S_{k},\ldots,S_n$.  For these states the minimal probability
  of moving towards the attractor is $p_{k+1}\geq p_k$.  Therefore we can set
  $p_i=p_{k}$ for $i\in\lbrace k+1,\ldots,n\rbrace$ to get an upper bound
  on the return time.  By reindexing the states we obtain the model $\tilde
  M=\model((p_{k})_{1\leq i \leq n-k})$ and, because of the truncation and the
  decrease of probabilities, $\tilde H_1$ is an upper bound on $H_{k+1}$, where
  $\tilde H_1$ is the expected number of steps to move from state $S_1$ to
  $S_0$ in model $\tilde M$.  In $\tilde M$ we have the fixed probabilities
  $p_{k}$ and can therefore apply Theorem \ref{thm:returntime} to determine
  $\tilde H_1$.  Therefore
  \begin{equation*}
    H_{k+1}\leq \tilde H_1
    =\frac{1-2p_{k}\left(\frac{1-p_{k}}{p_{k}}\right)^{n-k}}{2p_{k}-1}
    \leq\frac{1}{2p_{k} -1}
    =\frac{A(n)}{A(k)}\enspace.
  \end{equation*}
  Altogether we have $H_1\leq 2k+\frac{A(n)}{A(k)}$.
  With $k=n^{d/(d+1)}$, where $d$ is the degree of $A$, we get
  \begin{align*}
    H_1&
    \leq 2n^{\frac{d}{d+1}}+\frac{A(n)}{A(n^{d/(d+1)})}
    =\Theta(n^{\frac{d}{d+1}})+\frac{\Theta(n^d)}{\Theta((n^{d/(d+1)})^d)}
    \\&=\Theta(n^{\frac{d}{d+1}})+\Theta(n^{d-d^2/(d+1)})
    =\Theta(n^{\frac{d}{d+1}})\enspace,
  \end{align*}
  which certifies that $H_1=O(n^{d/(d+1)})$.
  By using Equation~\eqref{eq:partialunfold}
  we have the following lower bound on $H_1$:
  \begin{align*}
    H_1&\geq\sum_{i=1}^{k} \frac{1}{p_i} \prod_{j=1}^{i-1} \frac{1-p_j}{p_j}
    \geq\sum_{i=1}^{k}  \prod_{j=1}^{i-1} \frac{1-p_j}{p_j}\\
    \intertext{\centering{Note: $\prod_{j=1}^{i-1} \frac{1-p_j}{p_j}$ is monotonically decreasing as $p_j\geq 1/2$.}}
    &\geq    k\prod_{j=1}^{k-1} \frac{1-p_j}{p_j}
     =    k\prod_{j=1}^{k-1} \left(1-\frac{2\cdot A(j)}{A(n)+A(j)}\right)
     \\&\geq k\left(1-\sum_{j=1}^{k-1} \frac{2\cdot A(j)}{A(n)+A(j)}\right)
    \geq k\left(1-\sum_{j=1}^{k-1} \frac{2\cdot A(j)}{A(n)}\right)\\
    \intertext{\centering{Note: $A(j)$ is non-decreasing.}}
    &\geq k\left(1-k \cdot\frac{2\cdot A(k)}{A(n)}\right)
  \end{align*}
  As this equation holds for any $k\in\lbrace 1,\ldots,n\rbrace$ we can choose $k=g\cdot n^z$ with a not yet fixed constant $g\in(0,1)$ and $z=d/(d+1)$, $z\in(0,1)$. Please note that $g\cdot n^z$ tends to infinity if $n$ tends to infinity and therefore asymptotic expressions can also be applied if $g\cdot n^z$ is the argument. Please also note that $k$ has to be an integer but errors can be captured by some $\Theta(1)$ expressions. Substituting $k$ by $g\cdot n^z$ in the previous inequality results in 
  \begin{align*}
    H_1&\geq    \lceil g\cdot n^{z}\rceil\left(1-\lceil g\cdot n^{z}\rceil \cdot\frac{2\cdot A\left(\lceil g\cdot n^z\rceil\right)}{A(n)}\right)
	\\&=    g\cdot n^{z}\cdot\THETA(1)\left(1-g\cdot n^{z} \cdot \THETA(1)\cdot\frac{\THETA\left((g\cdot n^z)^d\right)}{\THETA(n^d)}\right)\\&
	=    g\cdot n^{z}\cdot \THETA(1)\left(1-g^{d+1}\cdot n^{z+d\cdot z-d}\cdot \THETA(1)\right)\\
    \intertext{\centering{Note: $z+d\cdot z-d=z\cdot(d+1)-d=d-d=0$.}}
    &
	=    g\cdot n^{z} \cdot \THETA(1)\left(1- g^{d+1}\cdot \THETA(1)\right)\\
	\intertext{Note: The last $\THETA(1)$ can be bounded from above by some constant $c_{\THETA}$ for large $n$ (by definition of $\THETA$).
Choose $g=\sqrt[d+1]{1/(2\cdot c_{\THETA})}$.
Then the expression in parentheses of the last term $\left(1- g^{d+1}\cdot \THETA(1)\right)$ may be negative for small $n$ but is at least $1/2$ for large $n$ which implies that it is in $\Omega(1)$.}
    &\geq g\cdot n^{z}\cdot\Omega\left(1\right)
    =\Omega\left(n^{\frac{d}{d+1}}\right)\enspace.
  \end{align*}
\end{proof}
  It may be verified that
  $$
  g=\sqrt[d+1]{\frac{\liminf\limits_{n\rightarrow\infty}(A(n)/n^d)}{\limsup\limits_{n\rightarrow\infty}(A(n)/n^d)}\cdot \frac{1}{4}\cdot (1-\varepsilon)}\enspace ,
  $$
  is a suitable choice for any $0 < \varepsilon < 1$, to replace $k$ by $g\cdot
  n^z$ to obtain the previous inequalities.  To see this, note that
  the first fraction of $\liminf$ and $\limsup$ counterbalances the fluctuation
  of $A(k)/A(n)$ relative to its $\Theta$-bound.  Furthermore, the factor
  $1/4$ compensates the factor $2$ which is hidden by $\THETA$ and supplies the
  desired factor of $1/2$ to ensure that the expression in parentheses is
  positive and at least $1/2$ for large $n$.
  Finally, the factor of $1 - \varepsilon$ is needed for some
  tolerance, because without it $g$ is only sufficiently small in the limit\footnote{Similar as $(n+1)/n$ has the limit one and there is no $n$ such that $(n+1)/n\leq 1$ but for any $\varepsilon\in(0,1)$ and large $n$ we have $(1-\varepsilon)(n+1)/n\leq 1$.} but
  not necessarily for large $n$.

This theorem is a generalization of Lemma~\ref{lemma:sqrt_return} and implies the $\THETA$-bound stated in that lemma by using $A(n)=n$.

\subsection{Bounds by Integration}
\label{subsec:integrationbounds}

Reformulations \eqref{eq:partialunfold} and \eqref{eq:recclosed} of the
recurrence~\eqref{eq:recmarkov} given in Section~\ref{sec:reformulation} do not yield closed-form
expressions of $H_1$, the expected number of steps it takes to return to the
attractor after an unsuccessful attempt to improve the current best solution.
In this section we derive closed-form expressions for $H_1$.
In order to get rid of the sums and products in
equations~\eqref{eq:partialunfold} and~\eqref{eq:recclosed}, we use the following
standard approaches.
First, sums may be approximated by their integral. This approach works quite
well if the summand/integrand is monotonic, which is true in our case.
Second, products can be transformed to integrals by reformulating the product
by the exponential function of the sum of logarithms.
This approach is an extension to \cite{MRSSW:17} as it is not present there at all.

\begin{theorem}
\label{thm:boundsbyintegration}
Let $M=\model((p(i))_{1\leq i\leq n})$ and $p:[0,n]\rightarrow(0,1]$ be a
non-decreasing function assigning the probabilities in the model, then
\begin{align*}
H_1=&\OMEGA\left(\mathrm{base}(p,n)^n
\right)\enspace,
\\
H_1=&O\left(
n\cdot \mathrm{base}(p,n)^n\right)\enspace\text{and}\\
H_1=&\Theta^*\left(\mathrm{base}(p,n)^n
\right)\enspace,\text{ where}
\end{align*}
$$
\mathrm{base}(p,n)=\sup_{k\in[0,n]}\exp \left(\int_0^{\frac{ k}n}\ln\left(\frac{1-p(n\cdot x)}{p(n\cdot x)}\right)\mathrm{d}x\right)\enspace.
$$
The integral in $\mathrm{base}(p,n)$ is maximized by $k=\inf\lbrace x\mid x\in[0,n]\wedge p(x)\geq 1/2\rbrace$ or $k=n$ if the infimum is taken on the empty set.
\end{theorem}

Please note that we use $\Theta^*$-bounds in the sense that polynomial factors can be omitted in $\Theta$-bounds. This is a similar notation as in the more common use case of $O^*$ where polynomial factors can also be omitted for upper bounds.
\begin{proof}
$p(i)$ is non-decreasing in $i$ and has values in $]0,1]$ and therefore
$\frac{1-p(i)}{p(i)}$ and also $\tau(i):=\ln\left(\frac{1-p(i)}{p(i)}\right)$ are non-increasing as the numerator is non-increasing and
the denominator is non-decreasing.
In the following series of equations, let $k \in [0, n)$.
Using Equation (\ref{eq:partialunfold}), we obtain
\begin{eqnarray*}
H_1\!\!\! &=& \sum_{i=1}^{n-1}\left(\frac1{p(i)}\cdot\prod_{j=1}^{i-1}\frac{1-p(j)}{p(j)}\right)+H_n\cdot\prod_{j=1}^{n-1}\frac{1-p(j)}{p(j)}
\ge \sum_{i=1}^{n-1}\left(\prod_{j=1}^{i-1}\frac{1-p(j)}{p(j)}\right)
\\&\ge&\prod_{j=1}^{\lfloor k\rfloor-1}\frac{1-p(j)}{p(j)}
\\&&    =\exp\left(\sum_{j=1}^{\lfloor k\rfloor-1}\tau(j)\right)\\
    \lefteqn{\text{Note: }\tau(j)=\ln\left((1-p(j))/p(j)\right)\text{ is non-increasing.}}\allowdisplaybreaks\\
&\ge& \exp\left(\int_1^{\lfloor k\rfloor}\tau(x)\mathrm{d}x\right)
= \exp \left(n\int_\frac1n^{\frac{\lfloor k\rfloor}n}\tau(n\cdot x)\mathrm{d}x\right)\allowdisplaybreaks\\
&=& \exp \left(n\cdot\left(\int_0^{\frac{ k}n}\ln\left(\frac{1-p(n\cdot x)}{p(n\cdot x)}\right)\mathrm{d}x\right.\right.\\
&&\phantom{\exp\left(n\cdot\left(\right.\right.}\left.\left.
- \int_0^{\frac{ 1}n}\ln\left(\frac{1-p(n\cdot x)}{p(n\cdot x)}\right)\mathrm{d}x
- \int_{\frac{\lfloor k\rfloor}n}^{\frac kn}\ln\left(\frac{1-p(n\cdot x)}{p(n\cdot x)}\right)\mathrm{d}x\right)\right)\allowdisplaybreaks\\
&\ge& \exp \left(n\cdot\left(\int_0^{\frac{ k}n}\ln\left(\frac{1-p(n\cdot x)}{p(n\cdot x)}\right)\mathrm{d}x
-\frac{2}{n}\ln\left(\frac{1-p(0)}{p(0)}\right)\right)\right)\allowdisplaybreaks\\
&=& \left(\frac{p(0)}{1-p(0)}\right)^2\exp \left(\int_0^{\frac{ k}n}\ln\left(\frac{1-p(n\cdot x)}{p(n\cdot x)}\right)\mathrm{d}x\right)^n\enspace.
\end{eqnarray*}
As $k$ can be chosen arbitrarily we get the claimed lower bound for $H_1$, because $\frac{p(0)}{1-p(0)}$ is a constant. As any integral is a continuous function also the whole expression in the supremum is a continuous function and therefore $k=n$ can be allowed in the supremum without changing the value.
Quite similar steps lead to the upper bound for $H_1$. We start with
Equation (\ref{eq:recclosed}).
\begin{align*}
H_1 &= \sum_{i=1}^{n}\left(\frac1{p(i)}\cdot\prod_{j=1}^{i-1}\frac{1-p(j)}{p(j)}\right)-\prod_{j=1}^{n}\frac{1-p(j)}{p(j)}
\leq\sum_{i=1}^{n}\left(\frac{1}{p(1)}\cdot\prod_{j=1}^{i-1}\frac{1-p(j)}{p(j)}\right)\\
&\leq\frac{n}{p(1)}\cdot \max_{i\in\lbrace 1,\ldots,n\rbrace}\left(\prod_{j=1}^{i-1}\frac{1-p(j)}{p(j)}\right)
=\frac{n}{p(1)}\cdot \max_{i\in\lbrace 1,\ldots,n\rbrace}\exp\left(\sum_{j=1}^{i-1}\ln\left(\frac{1-p(j)}{p(j)}\right)\right)\\
&\leq\frac{n}{p(1)}\cdot \max_{i\in\lbrace 1,\ldots,n\rbrace}\exp\left(\int_{0}^{i-1}\ln\left(\frac{1-p(x)}{p(x)}\right)\mathrm{d}x\right)
\\&
\overset{\mathclap{k=i-1}}{\leq}
\quad
\frac{n}{p(1)}\cdot \sup_{k\in[0,n]}\exp\left(\int_{0}^{k}\ln\left(\frac{1-p(x)}{p(x)}\right)\mathrm{d}x\right)\\
&\leq\frac{n}{p(1)}\cdot \sup_{k\in[0,n]}\exp\left(\int_{0}^{\frac{k}{n}}\ln\left(\frac{1-p(n\cdot x)}{p(n\cdot x)}\right)\mathrm{d}x\right)^n\enspace.\\
\end{align*}
This proves the claimed upper bound and as the base of the exponential part is equal for upper and lower bound we obtain the claimed $\Theta^*$ bound.

The logarithm in $\mathrm{base}(p,n)$ is positive as long as $1-p(n\cdot x)\geq
p(n\cdot x)$. Therefore the integral is maximized if we use the smallest possible $k$ (the infimum) which satisfies the condition $1-p(k)\leq p(k)\Leftrightarrow p(k)\geq\frac{1}{2}$.
\end{proof}
$p(n\cdot x)$ can in most cases be tightly bounded by a value independent of
$n$. This is the case if for example $p(i)=c+(1-c)i/n$, which we have for the
model solving \onemax by \onepso. The $k$ which maximizes the integral in the
expression of $\mathrm{base}(p,n)$ is usually obtained by solving the simple
equation $p(k)=1/2$.

Therefore the integral can be evaluated and the base of the exponential
part of the runtime can be determined.

\subsection{Variance of the Improvement Time}
\label{subsec:varianceAnalysis}

We show that the standard deviation of the return time is in the same order as the return
time. Therefore in experiments the average of such return times can be
measured such that a small relative error can be achieved.
Also the variance of $W_i$, the number of steps needed to move from state $S_i$ to state $S_{i-1}$, can be computed recursively.
Let $V_i:=\Var[W_i]$ be the variance of $W_i$.

To evaluate this variance we need the expectation and variance of a random
variable which is the sum of random variables where the number of summed up
random variables is also a random variable. Such random variables appear in the
Galton-Watson process (see \cite{durrett2010probability}) from one generation to the next generation.
\begin{lemma}
  \label{lem:composedrandomvariablevariance}
  Let $T$ be a random variable with non-negative integer values and let
  $(Y_i)_{i\in \N}$ be independent identically distributed random variables
  $Y_i\sim Y$ which are also independent of $T$. Additionally let
  $Z=\sum_{i=1}^{T}Y_i$.
  Then $\E[Z]=\E[T]\cdot\E[Y]$ and $\Var[Z]=\E[T]\cdot\Var[Y]+E[Y]^2\cdot\Var[T]$.
\end{lemma}
The statement on expected values is also known as Wald's equation and the
statement on the variance is known as the Blackwell-Girshick equation.  The
Blackwell-Girshick equation can be obtained by application of the law of total
variance:
\begin{align*}
  \Var[Z]&=\E[\Var[Z|T]]+\Var[\E[Z|T]]=\E[T\Var[Y]]+\Var[T\E[Y]]\\
  &=\E[T]\Var[Y]+\E[Y]^2\Var[T]\enspace.
\end{align*}

Also $W_i$ can be specified as a sum of random variables where the number of
summed up random variables is also a random variable. If we are currently in
state $S_i$ we have some success probability to move to $S_{i-1}$ in the next
iteration.  Therefore the number of trials in $S_i$ until we move to $S_{i-1}$
follows a geometric distribution. In case of failure we move to $S_{i+1}$ and
need additional $W_{i+1}$ steps until we can make our next attempt to move to
$S_{i-1}$.
Therefore
$$
W_i=\sum_{j=1}^{T-1}(\tilde W_{i+1,j}+1) + 1\enspace,
$$
where $T$ is a random variable distributed according to a geometric distribution with
success probability equal to the probability of moving to $S_{i-1}$ from $S_i$
and each $\tilde W_{i+1,j}$ is an independent copy of $W_{i+1}$.
\begin{theorem}
  \begin{align}
    \Var[W_i]=V_i&=\frac{1-p_i}{p_i}\cdot V_{i+1}+\frac{1-p_i}{p_i^2}\cdot(H_{i+1}+1)^2 \nonumber
                 \\&=\frac{1-p_i}{p_i}\cdot V_{i+1}+\frac{1}{1-p_i}\cdot(H_{i}-1)^2\enspace,&1\leq i<n\label{eq:varrec}\\
    \Var[W_n]=V_n&=0\enspace,\label{eq:varinitial}
  \end{align}
  where $p_i$ is the probability of moving to $S_{i-1}$ from $S_i$.
\end{theorem}
\begin{proof}
  $W_n=\E[W_n]=H_n=1\Rightarrow \Var[W_n]=V_n=\Var[1]=0$.\\
  Let $T$ be a random variable distributed according to a geometric distribution with
  success probability $p_i$ and let all $(\tilde W_{i+1,j})_{y\in\N}$ be independent
  copies of $W_{i+1}$.
  \begin{align*}
    \Var[W_i]&
    =\Var[W_i-1]
    =\Var\left[ \sum_{j=1}^{T-1}(\tilde W_{i+1,j}+1) \right]
    \\&\overset{\mathclap{\text{Lem.~\ref{lem:composedrandomvariablevariance}}}}{=}\quad\E[T-1]\cdot\Var[W_{i+1}+1]+\E[W_{i+1}+1]^2\cdot\Var[T-1]\\
    &=\frac{1-p_i}{p_i}\cdot \Var[W_{i+1}]+(\E[W_{i+1}]+1)^2\cdot\Var[T]
    \\&=\frac{1-p_i}{p_i}\cdot V_{i+1}+\frac{1-p_i}{p_i^2}\cdot(H_{i+1}+1)^2
  \end{align*}
  Finally the rightmost expression of Equation~\eqref{eq:varrec} is obtained by
  replacing $H_{i+1}$ according to Equation~\eqref{eq:recreform}.
\end{proof}
Therefore one can evaluate $H_i$ by
Equations~\eqref{eq:recreform} and \eqref{eq:recinitial} and then one can evaluate
$V_i$ by Equations~\eqref{eq:varrec} and \eqref{eq:varinitial}.

Please note that $V_i$ will always be in the same order as $H_i^2$. If
$(1-p_i)/p_i$ is less than one then the recursively needed values of $V_{j}$
for $j>i$ become less important and we have mainly $H_i^2$ and if $(1-p_i)/p_i$
is greater than one then $H_i^2$ is growing by at least $((1-p_i)/p_i)^2$ (see
Equation~\eqref{eq:recreform}) which is the square of the growing factor of
$V_i$.

As $V_i$ is in the same order as $H_i^2$ we obtain by an arithmetic average
of $T$ evaluations of $W_i$ a relative error of approximately $1/\sqrt{T}$.
This is indeed a relevant statistic if evaluations are performed and is consolidated in the following corollary.

\begin{corollary}
  Let $\tilde W_{i,j}\sim W_i$ be independent random variables. Then
  $$
  \E\left[\frac{\vert \sum_{j=1}^{T}\frac{\tilde W_{i,j}}{T} -H_i\vert}{H_i}\right]
  =O\left(\frac{1}{\sqrt{T}}\right)\enspace.
  $$
\end{corollary}

\section{Runtime Analysis of \onepso and \dpso}
\label{sec:runtime}

As mentioned above, we present a runtime analysis of \onepso for two combinatorial problems, the
sorting problem and \onemax. Our analysis is based on the \emph{fitness level
method}~\cite{Wegener:02}, in particular its application to the runtime
analysis of a \oneea for the sorting problem in~\cite{STW:04}. Consider a
(discrete) search space $X$ and an objective function $f:X \rightarrow \real$,
where $f$ assigns $m$ distinct values $f_1 < f_2 < \ldots < f_m$ on $X$. Let
$S_i \subseteq X$ be the set of solutions with value $f_i$. Assuming
that some algorithm \algo optimizing $f$ on $X$ leaves fitness level $i$ at
most once then the expected runtime of \algo is bounded from above by $\sum_{i=1}^m
{1}/{s_i}$, where $s_i$ is a lower bound on the probability of \algo leaving
$S_i$. The method has also been applied successfully, e.\,g., in~\cite{SW:10} to
obtain bounds on the expected runtime of a binary PSO proposed
in~\cite{KE:97}. 

\subsection{Upper Bounds on the Expected Optimization Time}
\label{sec:runtime:ub}

Similar to~\cite{STW:04,SW:10}, we use the fitness-level method to prove upper
bounds on the expected optimization time of the \onepso for sorting and
\onemax.  In contrast to the former, we allow non-improving solutions and
return to the attractor as often as needed in order to sample a neighbor of the
attractor that belongs to a better fitness level. Therefore, the time needed to
return to the attractor contributes a multiplicative term to the expected
optimization time, which depends on the choice of the algorithm parameter $c$.

We first consider the sorting problem. The structure of the search space of
the sorting problem has been discussed already in~\cite{STW:04} and a detailed
analysis of its fitness levels is provided in~\cite{MRSSW:17}.
In the following lemma we bound the transition probabilities for the
Markov model for the sorting problem. This allows us to bound the runtime of
\onepso for the sorting problem later on.

\begin{lemma}
  For the sorting problem on $n$ items, $c = 1/2$ and $x\in X_i$, the probability $p_x$
  that \onepso moves from $x$ to an element in $X_{i-1}$ is bounded from below by $p_i =
  \frac{1}{2} (1+ i/\binom{n}{2})$. Furthermore, this bound is tight.
  \label{thm:probabilities}
\end{lemma}
\begin{proof}
  The lower bound $p_i$ on $p_x$ can be obtained by
  \[
	p_x
	= \left(c+(1-c)\frac{\vert \neighborhood(x)\cap X_{i-1}\vert}{\vert\neighborhood(x)\vert}\right)
	= \frac{1}{2}\left(1+\frac{\vert \neighborhood(x)\cap X_{i-1}\vert}{\binom{n}{2}}\right) \geq p_i
	\enspace.
  \]
  To show the above inequality, consider the attractor $a$ and a permutation
  $\tau$ such that $x \circ \tau = a$.
  For each cycle of length $k$ of $\tau$, exactly $k-1$
  transpositions are needed to adjust the elements in this cycle and there are
  $\binom{k}{2}\geq k-1$ transpositions which decrease the transposition
  distance to the attractor $a$. Therefore the number of ways to decrease the
  transposition distance to $a$ is bounded from below by the transposition
  distance to $a$. Hence, we have $\vert \neighborhood(x)\cap X_{i-1}\vert \geq i$.  
  
  The lower bound is tight as it appears if only cycles
  of length two (or one) appear. In~\cite[Sec.~4]{MRSSW:17} a more detailed discussion on improvement
  probabilities can be found.
\end{proof}
Using Lemma~\ref{thm:probabilities} we prove the following bounds on  the
expected optimization time $T_{\rm sort}(n)$ required by \onepso for sorting
$n$ items by transpositions.
\begin{theorem}\textbf{\emph{\cite[Thm.~13]{MRSSW:17}}}
  \label{thm:onepso:sorting}
  The expected optimization time $T_{\rm sort}(n)$ of the \onepso sorting $n$ items is bounded from above by
  \[
	T_{\rm sort}(n) = 
	\begin{cases}
	  O(n^2 \log n )	&	\text{if $c \in (\frac{1}{2}, 1]$} \\
	  O(n^{{3}} \log n)	&	\text{if $c = \frac{1}{2}$} \\
	  O\left(\left( \frac{1-c}{c} \right)^n\cdot n^2\log n\right)	&	\text{if $c \in (0, \frac{1}{2})$\enspace.}
	\end{cases}
  \]
\end{theorem}
See Figure~\ref{fig:alphabeta} for a visualization of $\frac{1-c}{c}$.

\begin{proof}
  Consider the situation that the attractor has just been updated. Whenever the
  \onepso fails to update the attractor in the next iteration it will take in
  expectation $H_1$ iterations until the attractor is reached again and then it
  is improved with probability at least $i/\binom{n}{2}$. Again, if the \onepso
  fails to improve the attractor we have to wait $H_1$ steps, and so on. Since
  we do not consider the case that the attractor has been improved meanwhile,
  the general fitness level method yields an expected runtime of at most
  $\sum_{i=1}^n((H_1+1)(1/s_i-1)+1) = H_1 \cdot O(n^2 \log n)$.

  We now bound the expected return time $H_1$. Let $c \in (\frac{1}{2}, 1]$ and
  recall that $p_i$ is the probability of moving from state $S_i$ to state
  $S_{i-1}$. Then $1 \geq p_i > c > \frac{1}{2}$. Then the expression for $H_1$
  given in Theorem~\ref{thm:returntime} is bounded from above by the constant
  $1/(2c-1)$, so $T_{\rm sort}(n)=O(n^2\log n)$.
  Now let $c = \frac{1}{2}$, so $p_i \geq \frac{1}{2} (1+ i/\binom{n}{2})$ by
  Lemma~\ref{thm:probabilities}. Then, by Corollary~\ref{cor:returntime}, we
  have $H_1 = O(n)$, so $T_{\rm sort}(n)=O(n^3\log n)$. Finally, let $c \in (0, \frac{1}{2})$. Then $p_i > c > 0$,
  and by Theorem~\ref{thm:returntime}, $H_1$ is bounded from above by
  \[
	H_1 \leq \frac{2c}{1-2c}\left( \frac{1-c}{c} \right)^n = O\left(\left(\frac{1-c}{c}\right)^n\right)\enspace,
  \]
  so $T_{\rm sort}(n)=O\left(\left( \frac{1-c}{c} \right)^n\cdot n^2\log n\right)$.
\end{proof}

For $c = 0$, \onepso always moves to a uniformly drawn
adjacent solution.  Hence, the algorithm just behaves like a random walk on the
search space. Hence, in this case, $T_{\rm sort}(n)$ is the expected number of
transpositions that need to be applied to a permutation in order to obtain a
given permutation. We conjecture that $T_{\rm sort}(n)$ has the following
asymptotic behavior and provide theoretical evidence for this conjecture in
the Appendix~\ref{subsec:randomwalkconjecture}.
\begin{conjecture}
  \label{conject:random_walk}
  $T_{\rm sort}(n)\sim n!$ if $c=0$.
\end{conjecture}
Please note that the conjecture is actually only a conjecture on the upper
bound as Theorem~\ref{thm:onepso:sortinglb} supplies a proof that $T_{\rm
sort}(n)=\Omega(n!)$ if $c=0$.

Using a similar approach as in Theorem~\ref{thm:onepso:sorting}, we now bound
the expected optimization time $T_{\onemax}(n)$ of \onepso for \onemax.

\begin{theorem}
  The expected optimization time $T_{\onemax}(n)$ of the \onepso solving
  \onemax is bounded from above by
  \[
	T_{\onemax}(n) =
	\begin{cases}
	  O(n \log n)	&	\text{if $c \in (\frac{1}{2}, 1]$},\\
	  O(n^\frac{3}{2} \log n)	&	\text{if $c = \frac{1}{2}$},\\
	  O\left(\beta(c)^n\cdot n^2\log n\right)	&	\text{if $c \in (0, \frac{1}{2})$, and}\\
	  O(2^n)		&	\text{if $c = 0$\enspace.}
	\end{cases}
  \]
  where $\beta(c)=2^{{1}/({1-c})}\cdot (1-c)\cdot c^{{c}/({1-c})}\enspace$.
  \label{thm:onepso:onemax}
\end{theorem}
See Figure~\ref{fig:alphabeta} for a visualization of $\beta(c)$.
\begin{proof}
  The argument is along the lines of the proof of
  Theorem~\ref{thm:onepso:sorting}. We observe that on fitness level $0 \leq i
  \leq n$ there are $i$ bit flips that increase the number of ones in the
  current solution. Therefore, $s_i = i/n$ and the fitness level method yields
  an expected runtime of at most $\sum_{i=1}^n (H_1 + 1)(1/s_i-1)+1 =
  H_1 \cdot O(n \log n)$. The bounds on $H_1$ for $c > \frac{1}{2}$ are as
  in the proof of Theorem~\ref{thm:onepso:sorting}. For $c = \frac{1}{2}$ we
  invoke Lemma~\ref{lemma:sqrt_return} and have $H_1 = O(\sqrt{n})$. For $c <\frac{1}{2}$ we use Theorem~\ref{thm:boundsbyintegration}.
  The probabilities in the Markov model for $H_1$ are $p_i=c+(1-c)i/n$ which can be continuously extended to the non-decreasing function $p(i)=c+(1-c)i/n$.
  Here $k=n\cdot\frac{1-2c}{2(1-c)}$ solves the equation $p(k)=\frac{1}{2}$.
  Hence, we need the value of
  \begin{align}
& \mathrm{base}(p,n)
  =\exp\left(\int_0^\frac{1-2c}{2(1-c)}\ln\left(\frac {1-c-(1-c)\cdot x}{c+(1-c)\cdot x}\right)\mathrm{d}x\right)\nonumber\allowdisplaybreaks\\
&  =\exp\left(\int_0^\frac{1-2c}{2(1-c)}\left(\ln\left({1- x}\right)-\ln\left({\frac{c}{1-c}+x}\right)\right)\mathrm{d}x\right)\nonumber\\
& =\exp\left((x-1)\ln(1-x)-
  \left.\left(\frac{c}{1-c}+x\right)\ln\left(\frac{c}{1-c}+x\right)\right\vert^\frac{1-2c}{2(1-c)}_0\right)\nonumber\\
&  =2^{{1}/({1-c})}\cdot (1-c)\cdot c^{{c}/({1-c})}=\beta(c)\enspace.\label{eq:calculatebaseonemax}
  \end{align}
  Now Theorem~\ref{thm:boundsbyintegration} gives the upper bound
  $H_1=O(n\cdot\beta(c)^n)$.

  It remains to consider the case that $c=0$. The claimed bound
  on $T_{\rm \onemax}$ can be obtained by using the model
  $\model((\frac{i}{n})_{1\leq i\leq n})$. Each state represents the distance
  to the optimal point.  By Equation~\eqref{eq:recclosedk} we have
  \begin{align*}
    H_k
    &=\sum_{i=k}^n\frac{n}{i}\prod_{j=k}^{i-1}\frac{n-j}{j}
     =\prod_{j=1}^{k-1}\frac{j}{n-j}\sum_{i=k}^n\frac{n}{i}\prod_{j=1}^{i-1}\frac{n-j}{j}
    =\frac{1}{\binom{n-1}{k-1}}\sum_{i=k}^n\binom{n}{i}\leq\frac{2^n}{\binom{n-1}{k-1}}\enspace .
  \end{align*}
  The maximal expected time to reach the optimal point is the sum of all $H_k$:
  \begin{align*}
    T_{\onemax}(n)
    &\leq \sum_{k=1}^n H_k
     \leq \sum_{k=1}^n \frac{2^n}{\binom{n-1}{k-1}}
    =2^n\cdot\left(2+O\left(\frac{1}{n}\right)\right)
     =O(2^n)\enspace .
  \end{align*}
\end{proof}

We remark that the upper bounds given in Theorem~\ref{thm:onepso:onemax} for
$c\in[\frac{1}{2},1]$ were presented in \cite[Thm.~14]{MRSSW:17} and that the
upper bound for $c\in(0,\frac{1}{2})$ is newly obtained using the
\emph{bounds-by-integration} from Section~\ref{subsec:integrationbounds}
and the proof of the upper bound for $c=0$ is also new compared to
\cite{MRSSW:17}. Admittedly, the bound for $c=0$ is already available in the
context of randomized local search and can be found in \cite{Garnier1999}.
Furthermore, note that for $c=\frac{1}{2}$ it is not sufficient to use the
lower bound $p_i \geq p_1=\frac{1}{2} + \frac{1}{2n}$ in order to obtain the runtime
bound given in Theorem~\ref{thm:onepso:onemax}. 

By repeatedly running \onepso and applying Markov's inequality for the analysis, an optimal solution
is found with high probability so we have the following Corollary.
  \begin{corollary}
  If the \onepso is repeated $\lambda\cdot\log_2(n)$ times but each repetition
  is terminated after $2\cdot T(n)$ iterations, where $T(n)$ is the upper bound
  on the expected number of iterations to find the optimum specified in
  Theorems~\ref{thm:onepso:sorting} and~\ref{thm:onepso:onemax} with suitable
  constant factor, then \onepso finds the optimal solution with high
  probability. 
  \end{corollary}

\subsection{Lower Bounds via Indistinguishable States}
\label{sec:runtime:lb}

In this section we will provide lower bounds on the expected optimization time
of \onepso that almost match our upper bounds given in
Section~\ref{sec:runtime:ub}.  We will use the Markov model from
Section~\ref{sec:model} to obtain these lower bounds. The main difference to
the previous section is that we restrict our attention to the \emph{last}
improvement of the attractor, which dominates the runtime, both for sorting and
\onemax. 
We will introduce the useful notion of \emph{indistinguishability} of certain
states of a Markov chain.
Note that our lower bounds are significantly improved compared to
the conference version \cite{MRSSW:17} by using the newly introduced
bounds-by-integration from Section~\ref{subsec:integrationbounds}.

\subsubsection{Indistinguishable States}
We now introduce a notion of \emph{indistinguishability} of certain
states of a Markov chain already presented in~\cite{MRSSW:17}. We will later use this notion to prove lower bounds
on the expected optimization time of \onepso for sorting and \onemax as
follows: We show that the optimum is contained in a set $\hat{Y}$ of
indistinguishable states. Therefore, in expectation, the states $\hat{Y}$ have
to be visited $\Omega(|\hat{Y}|)$ times to hit the optimum with positive
constant probability.

\begin{definition}[Indistinguishable states]
\label{defi:symmetric}
Let $M$ be a Markov process with a finite set $Y$ of states  and let $\hat{Y}\subseteq Y$.
Furthermore, let $(Z_i)_{i\geq0}$ be the sequence of visited states of $M$ and let $T=\min\lbrace t>0\mid Z_t\in\hat{Y}\rbrace$. Then
$\hat{Y}$ is called \emph{indistinguishable} with respect to $M$ if
\begin{enumerate}
\item \label{statement:equal_prob}the initial state $Z_0$ is uniformly
distributed over $\hat{Y}$, i.\,e., for all $y\in Y$:
\[
    \Pr[Z_0=y]=\1_{y\in\hat{Y}}/\vert\hat{Y}\vert =
    \begin{cases}
        1/\vert\hat{Y}\vert & \text{if }y\in\hat{Y}\\
        0 & \text{if }y\not\in\hat{Y} \enspace.
    \end{cases}
\]
\item \label{statement:symmetric_prob}and the probabilities to reach states in
$\hat{Y}$ from states in $\hat{Y}$ are symmetric, i.\,e., for all $y_1,y_2\in \hat{Y}$:
\[
	\Pr[Z_T=y_2\mid Z_0=y_1]=\Pr[Z_T=y_1\mid Z_0=y_2] \enspace.
\]
\end{enumerate}
\end{definition}

Now we can prove a lower bound on the expected time for finding a specific state.

\begin{theorem}
\label{thm:lower_bound}
Let $M$ be a Markov process as in Definition~\ref{defi:symmetric} and let
$\hat{Y}$ be indistinguishable with respect to $M$.  Let $h(M)$ be a
positive real value such that $\E[T]\geq h(M)$, then the expected time to reach
a fixed $y\in\hat{Y}$ is bounded below by
$h(M)\cdot\OMEGA(\vert\hat{Y}\vert)$.
\end{theorem}

\begin{proof}
Let $T_i$ be the stopping time when $\hat{Y}$ is visited the $i$-th time.
\[
    T_i=\min\lbrace t\geq 0 \mid \vert\lbrace k\mid 0\leq k \leq t \wedge Z_k\in\hat{Y}\rbrace\vert\geq i\rbrace \enspace.
\]
With Statement~\ref{statement:equal_prob} of Definition \ref{defi:symmetric}
$Z_0$ is uniformly distributed over $\hat{Y}$. Therefore $T_1=0$ and
$T_2=T$. Statement \ref{statement:symmetric_prob} of Definition
\ref{defi:symmetric} implies that
$\Pr[Z_{T_i}=y]=\1_{y\in\hat{Y}}/\vert\hat{Y}\vert$ for all $i\geq 1$
by the following induction. The base case for $i=1$ and $T_i=0$ is ensured by
the Statement~\ref{statement:equal_prob} of Definition \ref{defi:symmetric}.
The induction hypothesis is
$\Pr[Z_{T_{i-1}}=y]=\1_{y\in\hat{Y}}/\vert\hat{Y}\vert$. The
inductive step is verified by the following series of equations.
\begin{align*}
\Pr[Z_{T_i}=y]
&=
\sum_{\hat{y}\in\hat{Y}}\Pr[Z_{T_{i-1}}=\hat{y}]\cdot \Pr[Z_{T_i}=y\mid Z_{T_{i-1}}=\hat{y}]
\allowdisplaybreaks\\
\stackrel{\mathclap{\text{ind. hyp.}}}{=}\hspace*{1.5em}&
\sum_{\hat{y}\in\hat{Y}}1/\vert\hat{Y}\vert\cdot \Pr[Z_{T_i}=y\mid Z_{T_{i-1}}=\hat{y}]
\allowdisplaybreaks\\
\stackrel{\mathclap{\text{Def.\ref{defi:symmetric} St.\ref{statement:symmetric_prob}}}}{=}\hspace*{1.5em}&
1/\vert\hat{Y}\vert\cdot\sum_{\hat{y}\in\hat{Y}}\Pr[Z_{T_i}=\hat{y}\mid Z_{T_{i-1}}=y] 
=
1/\vert\hat{Y}\vert \enspace.
\end{align*}
It follows that for all $i>0$ the difference $T_{i+1}-T_i$ of two consecutive
stopping times has the same distribution as $T$ and also
\[
    \E[T_{i+1}-T_{i}]=\E[T]\geq h(M)\enspace.
\]
Now let $y\in\hat{Y}$ be fixed. The probability that $y$ is not reached within
the first $T_{\lfloor\vert\hat{Y}\vert/2\rfloor-1}$ steps is bounded from below
through union bound by
\[
  1-\Pr[Z_0=y]-\sum_{i=1}^{\lfloor\vert\hat{Y}\vert/2\rfloor-1}\Pr[Z_{T_i}=y]\geq1/2
\]
and therefore the expected time to reach the fixed $y\in\hat{Y}$ is
bounded from below by
\begin{align*}\frac{1}{2}\cdot
\E[T_{\lfloor\vert\hat{Y}\vert/2\rfloor-1}]
&=\frac{1}{2}\cdot\sum_{i=2}^{\lfloor\vert\hat{Y}\vert/2\rfloor-1}\E[T_i-T_{i-1}]\\
&\geq
\frac{1}{2}\cdot\sum_{i=2}^{\lfloor\vert\hat{Y}\vert/2\rfloor-1}h(M)=h(M)\cdot\OMEGA(\vert\hat{Y}\vert).
\end{align*}
\end{proof}

\subsubsection{Lower Bounds on the Expected Optimization Time for Sorting}

In this section we consider the sorting problem.  Our first goal is to provide lower
bounds on the expected return time to the attractor for the parameter choice
$c\in(0,\frac{1}{2})$.  
\begin{lemma}
\label{lemma:lower_bound_sort_small_c}
Let $c\in(0,\frac12)$. For the sorting problem on $n$ items, assume that the attractor has transposition distance one to the identity permutation. Then the expected return time $H_1$ to the attractor is bounded from below by $\OMEGA(\alpha(c)^n)$, where
$$
\alpha(c)=\left(\frac{1+\sqrt{\frac{1-2c}{2(1-c)}}}{1-\sqrt{\frac{1-2c}{2(1-c)}}}\right)\cdot\exp\left(
-2\sqrt{\frac{c}{1-c}}\arctan\left(\sqrt{\frac{1-2c}{2c}}\right)\right)\enspace.
$$
\end{lemma}
See Figure~\ref{fig:alphabeta} for a visualization of $\alpha(c)$.
\begin{proof}
The probability of decreasing the distance to the attractor in state $S_i$ can be bounded from above by 
$$
p_{i-1} \le c+(1-c)\cdot\frac{\binom{i}{2}}{\binom{n}{2}} = c+(1-c)\cdot\frac{i(i-1)}{n(n-1)} \le c+(1-c)\cdot\frac{i^2}{n^2}\enspace.
$$
We increase all indices by one such that $\tilde p_i=p_{i-1}$ such that we have $n$ states again.
Please note that $H_2=\Omega(H_1)$. This can be obtained by the following equations while using Equation~\ref{eq:recreform} and the fact that $H_1\geq \frac{1}{1-c}$ is true in this case
$$
H_2
=\frac{p_1}{1-p_1}H_1-\frac{1}{1-p_1}
\geq \frac{c}{1-c}H_1-\frac{1}{1-c-o(1)}=\Omega(H_1)\enspace.$$
We use Theorem~\ref{thm:boundsbyintegration} to get a lower bound on $H_1$ by using $p(i)=c+(1-c)\cdot\frac{i^2}{n^2}$.
Here $k=n\cdot \sqrt{\frac{1-2c}{2(1-c)}}$ maximizes the integral, because it solves the equation $p(k)=\frac{1}{2}$.
An application of Theorem~\ref{thm:boundsbyintegration} supplies
$$
H_1\!=\!\OMEGA\left(
\exp \left(\int_0^{\sqrt{\frac{1-2c}{2(1-c)}}}\ln\left(\frac{1-c-(1-c)x^2}{c+(1-c)x^2}\right)\mathrm{d}x\right)^{\!\!\!n}\right)\enspace .
$$
In the following we calculate the exact value of this integral.
The integrand can be converted to the expression
\begin{align*}
&\ln\left(\frac{1-c-(1-c)x^2}{c+(1-c)x^2}\right)
=\ln\left(\frac{1-x^2}{\frac{c}{1-c}+x^2}\right)
=\ln(1-x^2)-\ln\left(\frac{c}{1-c}+x^2\right)\enspace.
\end{align*}
The indefinite integral of $\ln(1-x^2)$ is
$$
x\cdot \ln(1-x^2)-2x+\ln\left(\frac{1+x}{1-x}\right)\enspace.
$$
It can be evaluated for values $x\in[0,1[$, but this is fine as $0\leq k/n<1$.
Furthermore the indefinite integral of $\ln\left(\frac{c}{1-c}+x^2\right)$ is
$$
x\cdot \ln\left(\frac{c}{1-c}+x^2\right)-2x+2\sqrt{\frac{c}{1-c}}\arctan\left({x}\cdot{\sqrt{\frac{1-c}{c}}}\right)\enspace,
$$
which can be evaluated for all values, because $\frac{c}{1-c}$ is positive.
The indefinite integral of the whole expression is obtained by the subtraction of both
\begin{align*}
&x\cdot \ln(1-x^2)+\ln\left(\frac{1+x}{1-x}\right)
-x\cdot \ln\left(\frac{c}{1-c}+x^2\right)
-2\sqrt{\frac{c}{1-c}}\arctan\left({x}\cdot{\sqrt{\frac{1-c}{c}}}\right)
\end{align*}
and evaluation of the bounds $k/n=\sqrt{\frac{1-2c}{2(1-c)}}$ and $0$ results in
\begin{align*}
  &\left[\sqrt{\frac{1-2c}{2(1-c)}}\cdot \ln\left(1-\frac{1-2c}{2(1-c)}\right)+\ln\left(\frac{1+\sqrt{\frac{1-2c}{2(1-c)}}}{1-\sqrt{\frac{1-2c}{2(1-c)}}}\right)\right.\allowdisplaybreaks[0]\\
&\phantom{\bigg[}
  -\sqrt{\frac{1-2c}{2(1-c)}}\cdot \ln\left(\frac{c}{1-c}+\frac{1-2c}{2(1-c)}\right)\allowdisplaybreaks[0]\\
&\phantom{\bigg[}\left.
-2\sqrt{\frac{c}{1-c}}\arctan\left(\frac{\sqrt{\frac{1-2c}{2(1-c)}}}{\sqrt{\frac{c}{1-c}}}\right)\right]-\left[0+\ln(1)-0-0\right]\\
=&\sqrt{\frac{1-2c}{2(1-c)}}\cdot \ln\left(\frac{1}{2(1-c)}\right)+\ln\left(\frac{1+\sqrt{\frac{1-2c}{2(1-c)}}}{1-\sqrt{\frac{1-2c}{2(1-c)}}}\right)
\\&-\sqrt{\frac{1-2c}{2(1-c)}}\cdot \ln\left(\frac{1}{2(1-c)}\right)
\\&-2\sqrt{\frac{c}{1-c}}\arctan\left(\sqrt{\frac{1-2c}{2c}}\right)\\
=
&\ln\left(\frac{1+\sqrt{\frac{1-2c}{2(1-c)}}}{1-\sqrt{\frac{1-2c}{2(1-c)}}}\right)
-2\sqrt{\frac{c}{1-c}}\arctan\left(\sqrt{\frac{1-2c}{2c}}\right)\enspace.
\end{align*}
An application of the $\exp$ function on this result gives the claimed lower bound.
\end{proof}

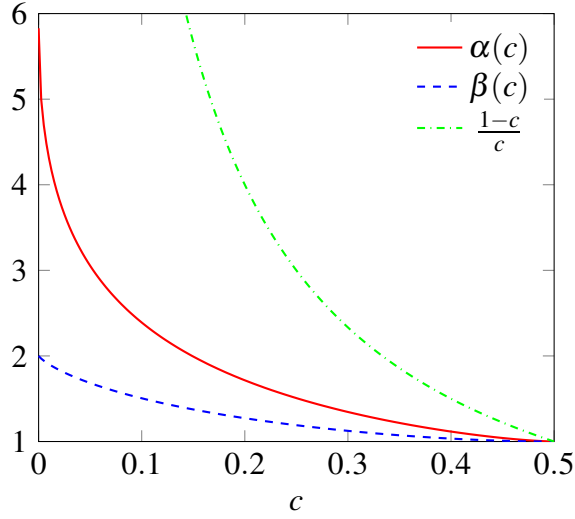
\begin{figure}[tb]
  \centering
  \begin{tikzpicture}[
  baseline,
  declare function={ arctanh(\x) = 0.5*ln((1+\x)/(1-\x));
                     arctan(\x)  = rad(atan(\x));}
  ]
	\begin{axis}[
		xlabel={$c$},
		ymin=1, ymax=6,
		xmin=0, xmax=0.5,
        samples=200,
		legend entries={$\alpha(c)$,$\beta(c)$,$\frac{1-c}{c}$},
        legend style={draw=none},
	  ]
	  \addplot[red,thick,domain=0:0.5]  {( (1+ sqrt((1-2*x)/(2*(1-x))))/(1- sqrt((1-2*x)/(2*(1-x))))*exp( -2*sqrt(x/(1-x))*arctan(sqrt((1-2*x)/(2*x)))))};
	  \addplot[dashed, thick, blue,domain=0:0.5] {( 2^(1/(1-x))*(1-x)*(x^(x/(1-x))))};
	  \addplot[dashdotted, thick, green,domain=0:0.5] {(1-x)/x};
	\end{axis}
  \end{tikzpicture}
  \caption{The functions $\alpha(c)$, $\beta(c)$ and $\frac{1-c}{c}$ for $c \in (0,\frac{1}{2})$}
  \label{fig:alphabeta}
\end{figure}

This lower bound is the best possible bound which can be achieved with this
model as the probability $p_i = c+(1-c)\cdot{\binom{i+1}{2}}/{\binom{n}{2}}$
actually appears at distance $i$ if the permutation transforming the current
position to the attractor consists of one cycle of length $i+1$ and the
remaining permutation consists of singleton cycles. For this improvement
probability the bound is $\Theta^*(\alpha(c)^n)$.

The following theorem supplies lower bounds on the expected optimization
time of \onepso on the sorting problem.

\begin{theorem}
  The expected optimization time $T_{\rm sort}(n)$ of the \onepso sorting $n$
  items is bounded from below by
  \[
	T_{\rm sort}(n) = 
	\begin{cases}
    \Omega(n^2)	&	\text{if }c \in (\frac{1}{2}, 1] \\
  \Omega(n^{\frac{8}{3}})	&	\text{if }c = \frac{1}{2} \\
\Omega \left(\alpha(c)^n\cdot n^2\right)	&	\text{if }c \in (0, \frac{1}{2})\\
	  \Omega\left(n!\right)	&	\text{if $c =0 \enspace$.}
	\end{cases}
  \]
  \label{thm:onepso:sortinglb}
\end{theorem}
\begin{proof}
  The situation where already the initial position is the optimum has probability $1/n!$.
  As $1-1/n!>1/2$ for $n\geq 2$ we have the same $\Omega$ bound if we ignore this case.
  In all other cases we can consider the situation that the attractor has just been updated to a solution
  that has distance one to the optimum. Without loss of generality, we assume
  that the attractor is the identity permutation and the optimum is the
  transposition $(0\,1)$. The number of steps required for the next (hence
  final) improvement of the attractor is a lower bound on the expected optimization
  time for the \onepso. We determine a lower bound on this number for various
  choices of~$c$.

  For all $c \in (0,1]$ we apply Theorem~\ref{thm:lower_bound}. We use all
  permutations as set of states $Y$ in the Markov process $M$. Let
  $\hat{Y}=X_1$ be the subset of states which are a single swap away from the
  attractor. Therefore the optimal solution is contained in $\hat{Y}$, but up
  to the point when the \onepso reaches the optimal solution it is
  indistinguishable from all other permutations in $\hat{Y}$. We will
  immediately prove that $\hat{Y}$ is actually indistinguishable with respect
  to $M$.  Initially the particle is situated on the attractor and after a
  single step it is situated at a permutation in $\hat{Y}$, where each
  permutation has equal probability. We use the permutation after the first
  step as the initial state of the Markov process $Z_0$ and all other $Z_i$ are
  the successive permutations. Therefore Statement~\ref{statement:equal_prob}
  of Definition~\ref{defi:symmetric} is fulfilled.  Let $T=\min\lbrace t>0\mid
  Z_t\in\hat{Y}\rbrace$ the stopping time of Theorem~\ref{thm:lower_bound}. For
  each sequence of states $Z_0,\ldots,Z_T$ there is a one to one mapping to a
  sequence $\tilde Z_0=Z_T,\tilde Z_1,\ldots,\tilde Z_{T-1},\tilde Z_{T}=Z_0$
  which has equal probability to appear. The sequence $\tilde Z_0,\ldots,\tilde
  Z_T$ is not the reversed sequence, because the forced steps would then lead
  to the wrong direction, but the sequence can be obtained by renaming the
  permutation indices. The renaming is possible because the permutations $Z_0$
  and $Z_T$ are both single swaps. As this one to one mapping exists also the
  Statement~\ref{statement:symmetric_prob} of Definition~\ref{defi:symmetric}
  is fulfilled. Finally we need a bound on the expectation of $T$. If we are in
  $X_1=\hat{Y}$ we can either go to the attractor by a forced move or random
  move and return to $X_1$ in the next step or we can go to $X_2$ by a random
  move and return to $X_1$ in expectation after $H_2$ steps.
  We have $\E[T]=\left(c+(1-c)/\binom{n}{2}\right)\cdot 2 +
  (1-c)\cdot\left(1-1/\binom{n}{2}\right)(1+H_2)=\Omega(H_2)=:h(M)$.
  Theorem~\ref{thm:lower_bound} provides the lower bound
  $\OMEGA(\vert\hat{Y}\vert\cdot H_2)$ for the runtime to find the fixed
  permutation $(0,1)\in\hat{Y}$ which is the optimal solution. From
  Equation~\ref{eq:recreform} we get $H_2=(p_1\cdot H_1-1)/(1-p_1)\geq(c\cdot
  H_1-1)/(1-c)$. As $H_1=\OMEGA(n^{2/3})$ for $c=\frac{1}{2}$ (see
  Theorem~\ref{thm:h1_quad_theta}) and $H_1=\OMEGA(\alpha(c)^n)$ for
  $c\in(0,\frac{1}{2})$ (see Lemma~\ref{lemma:lower_bound_sort_small_c}) also
  $H_2=\OMEGA(H_1)$ for $c\in (0,\frac{1}{2}]$ which results in the lower
  bounds $T_{\rm sort}(n)=\OMEGA(\vert\hat{Y}\vert\cdot
  H_1)=\OMEGA(\binom{n}{2}\cdot n^{2/3})=\OMEGA(n^{8/3})$ for $c=\frac{1}{2}$
  and $T_{\rm sort}(n)=\OMEGA(\vert\hat{Y}\vert\cdot
  H_1)=\OMEGA(\binom{n}{2}\cdot \alpha(c)^n)=\OMEGA(n^2\cdot \alpha(c)^n)$ for
  $c\in(0,\frac{1}{2})$. Trivially the return time to $X_1$ in $M$ can be
  bounded by $2$, which results in the lower bound $T_{\rm
  sort}(n)=\OMEGA(n^2)$ for the case $c\in(\frac{1}{2},1]$.

  The lower bound for $c=0$ can be
  derived directly from the indistinguishability property: Let $\hat{Y} =
  Y$. It is readily verified that the initial state is uniformly distributed
  over $\hat{Y}$. Furthermore, any $\hat{Y}$-$\hat{Y}$-path can be reversed and
  has the same probability to occur. Therefore,
  Condition~\ref{statement:symmetric_prob} of Definition~\ref{defi:symmetric}
  is satisfied and the lower runtime bound follows from
  Theorem~\ref{thm:lower_bound} by choosing $h(M) = 1$.
\end{proof}

Beside that formally proved lower bounds we conjecture the following lower
bounds on the expected optimization time of \onepso for sorting $n$ items.
\begin{conjecture}
  \label{con:sorting:lb}
  \[
	T_{\rm sort}(n) = 
	\begin{cases}
	  \Omega(n^2)	&	\text{if $c \in (\frac{1}{2}, 1]$} \\
	  \Omega(n^{3})	&	\text{if $c = \frac{1}{2}$} \\
      \Omega \left(\left(\frac{1-c}{c}\right)^n\cdot n^2\right)	&	\text{if $c \in (0, \frac{1}{2})$}\enspace.
	\end{cases}
  \]
\end{conjecture}

Note that these lower bounds differ from our upper bounds given in
Theorem~\ref{thm:onepso:sorting} only by a $\log$-factor. Evidence supporting this conjecture is given in
Appendix~\ref{subsec:approxOptTime}. We obtain our theoretical evidence by
considering the \emph{average} probability to move towards the attractor,
instead of upper and lower bounds as before.

\subsubsection{Lower Bounds on the Expected Optimization Time for \onemax}
First we provide a lower bound on the expected return time to the attractor.
\begin{lemma}
\label{lemma:lower_bound_1max_small_c}
Let $c\in(0,\frac{1}{2})$. For \onemax, assume that
the attractor has Hamming distance one to the optimum $1^n$. Then
the expected return time $H_1$ to the attractor is bounded from below by
$H_1=\OMEGA(\beta(c)^n)$, where
\[
\beta(c)=2^{{1}/({1-c})}\cdot (1-c)\cdot c^{{c}/({1-c})}\enspace.
\]
\end{lemma}
See Figure~\ref{fig:alphabeta} for a visualization of $\beta(c)$.
\begin{proof}
We use Theorem~\ref{thm:boundsbyintegration}. The value $\beta(c)$ is already calculated in Theorem~\ref{thm:onepso:onemax} Equation~\ref{eq:calculatebaseonemax}.
\end{proof}

This result enables us to prove lower bounds on $T_{\onemax}(n)$.
\begin{theorem}
  The expected optimization time $T_{\onemax}(n)$ of the \onepso for solving
  \onemax is bounded from below by
  \[
	T_{\onemax}(n) =
	\begin{cases}
	  \Omega(n \log n)	&	\text{if $c \in (\frac{1}{2}, 1]$}\\
	  \Omega(n^\frac{3}{2})	&	\text{if $c = \frac{1}{2}$}\\
	  \Omega\left(\beta(c)^n\cdot n\right)	&	\text{if $c \in (0, \frac{1}{2})$}\\
	  \Omega\left(2^n\right)	&	\text{if $c =0 \enspace$.}
	\end{cases}
  \]
  \label{thm:onepso:onemaxlb}
\end{theorem}
\begin{proof}
  First, let $c \in (\frac{1}{2}, 1]$. Then, with probability at least
  $\frac{1}{2}$, the initial solution contains at least $k = \lfloor n/2
  \rfloor = \Omega(n)$ zeros. Each zero is flipped to one with probability
  $1/n$ in a random move, and none of the $k$ entries is set to one in a move
  towards the attractor. The expected time required to sample the $k$ distinct
  bit flips is bounded from below by the expected time it takes to obtain all
  coupons in the following instance of the coupon collector's problem: There
  are $k$ coupons and each coupon is drawn independently with probability
  $1/k$. The expected time to obtain all coupons is $\Omega(k \log k)$~
  \cite[Section~5.4.1]{MU:05}. It follows that the expected optimization time is
  $\Omega(n \log n)$ as claimed.

  For $c \in (0,\frac{1}{2}]$ we use the same approach as in the proof of
  Theorem~\ref{thm:onepso:sortinglb}. Also here the event that the initial
  solution is optimal can be ignored. Consider the situation that the attractor
  has just been updated to a solution that has distance one to the optimum. We
  use the set of all bit strings as set of states $Y$ in the Markov process
  $M$. Let $\hat{Y}=X_1$ the subset of bit strings which is a single bit flip
  away from the attractor, hence $\hat{Y}$ contains the optimum.  $Z_i$ and $T$
  are instantiated as in the proof of Theorem~\ref{thm:onepso:sortinglb}.
  Therefore Statement~\ref{statement:equal_prob} of
  Definition~\ref{defi:symmetric} is fulfilled. Again for each sequence of
  states $Z_0,\ldots,Z_T$ we have a one to one mapping to a sequence $\tilde
  Z_0=Z_T,\tilde Z_1,\ldots,\tilde Z_{T-1},\tilde Z_{T}=Z_0$ which has equal
  probability to appear. This sequence is again obtained by renaming the
  indices plus some bit changes according to the shape of the attractor. Hence
  also Statement~\ref{statement:symmetric_prob} of
  Definition~\ref{defi:symmetric} is fulfilled. Hence $\hat{Y}$ is
  indistinguishable with respect to $M$. We obtain $\E[T]=\OMEGA(H_2)=:h(M)$.
  Theorem~\ref{thm:lower_bound} provides the lower bound
  $\OMEGA(\vert\hat{Y}\vert\cdot H_2)$ for the runtime to find the optimal
  solution. From Equation \eqref{eq:recreform} we get $H_2\geq(c\cdot H_1-1)/(1-c)$. As $H_1=\OMEGA(n^{1/2})$ for
  $c=\frac{1}{2}$ (see Theorem~\ref{thm:h1_lin_theta}) and
  $H_1=\OMEGA(\beta(c)^n)$ for $c\in(0,\frac{1}{2})$ (see
  Lemma~\ref{lemma:lower_bound_1max_small_c}) also $H_2=\OMEGA(H_1)$ for $c\in
  (0,\frac{1}{2}]$ which results in the lower bounds
  $T_{\onemax}(n)=\OMEGA(\vert\hat{Y}\vert\cdot H_1)=\OMEGA(n\cdot
  n^{1/2})=\OMEGA(n^{3/2})$ for $c=\frac{1}{2}$ and
  $T_{\onemax}(n)=\OMEGA(\vert\hat{Y}\vert\cdot H_1)=\OMEGA(n\cdot
  \beta(c)^n)=\OMEGA(n\cdot \beta(c)^n)$ for $c\in(0,\frac{1}{2})$.

  Again the lower bound for $c=0$ can be obtained by the indistinguishability
  property.  The proof for this case is identical to the corresponding part of
  Theorem~\ref{thm:onepso:sortinglb}.
\end{proof}

  Finally all runtime bounds claimed in Table~\ref{tab:summary} are justified
  and for this purpose all of the presented tools in Section~\ref{sec:tools}
  are used.

\subsection{Bounds on the Expected Optimization Time for \dpso}
\label{sec:dpsobounds}
  The upper bounds on the runtime of \onepso in Theorem
  \ref{thm:onepso:sorting} and Theorem \ref{thm:onepso:onemax} directly imply
  upper bounds for \dpso. Recall that we denote by $c$ the parameter of \onepso
  and by $T_\onemax(n)$ and $T_{\rm sort}(n)$ the expected optimization time of
  \onepso for \onemax and sorting, respectively. 
  \begin{corollary}
    Let $T'_{\onemax}(n)$ and $T'_{\rm sort}(n)$ be the expected optimization
    time of \dpso for \onemax and sorting, respectively. If $c = c_{glob}$,
    then $T'_{\onemax}(n) = O(P\cdot T_{\onemax}(n))$ and $T'_{\rm sort}(n)
    =O(P\cdot T_{\rm sort}(n))$, where $P$ is the number of particles.
  \end{corollary}
  \begin{proof}
    In each trial to improve the value of the global attractor at least the
    particle which updated the global attractor has its local attractor at the
    same position as the global attractor. This particle behaves exactly like
    the single particle in \onepso until the global attractor is improved.
    Therefore we have at most $P$ times more objective function evaluations than \onepso,
    where $P$ is the number of particles. 
  \end{proof}
  One can refine this result by again looking on return times to an attractor.

  If the global attractor equals the local attractor then this particle performs the same
  steps as all other particles having equal attractors. As all those particles
  perform optimization in parallel in expectation no additional objective
  function evaluations are made compared to the \onepso.

  For particles where the global attractor and the local attractor differ
  we can use previous arguments applied to the local attractor.
  With two different attractors alternating movements to the local
  and global attractor can cancel each other out.
  Therefore if (only) $c_{loc}$ is fixed then for the worst case we can assume only
  $c_{loc}$ as probability of moving towards the local attractor and
  $1-c_{loc}$ as probability of moving away from the local attractor.
  This enables us to use Theorem~\ref{thm:returntime} to calculate the
  expected time to reduce the distance to an attractor from one to zero.
  We denote the return time from Theorem~\ref{thm:returntime} as $\Psi(n,p)$ 
  \begin{equation*}
    \Psi(n,p):=\left.\begin{cases}
	  \dfrac{1-2p \left( \dfrac{1-p}{p} \right)^n }{2p-1}	&	\text{if $p \neq \frac{1}{2}$}\\
	  2n-1															 &   \text{otherwise}
  \end{cases}\right\rbrace
    =\begin{cases}
      \Theta(1)	&	\text{if $\frac{1}{2}<p\leq 1$}\\
      \Theta(n)	&	\text{if $p = \frac{1}{2}$}\\
      \Theta\left( \left( \frac{1-p}{p} \right)^n \right)	&	\text{if $0<p < \frac{1}{2}$}\enspace .
	\end{cases}
  \end{equation*}
  If the position equals the local attractor and consequently differs from the
  global attractor the probability for improving the local attractor can be
  bounded from below by a positive constant.
  E.\,g., for the problem \onemax this constant is $c_{glob}/2$ because for a move
  towards the global attractor for at least half of the differing bits the
  value of the global attractor equals the value of the optimal solution as the
  global attractor is at least as close to the optimum as the local attractor.
  Therefore the number of trials until the local attractor is improved is constant.
  As such an update occurs at most once for each particle and fitness level we
  obtain an additional summand of $O(\Psi(n,c_{loc})\cdot P\cdot n)$ instead of
  the factor $P$ for the problems \onemax and the sorting problem.

  In contrast to the upper bounds, the lower bounds for \onepso do not apply for \dpso for the
  following reason. The bottleneck used for the analysis of \onepso is the very last improvement
  step. However, \dpso may be faster at finding the last improvement because
  it may happen that the local and global attractor of a particle have both
  distance one to the optimum but are not equal. In this case, as described above, there is a
  constant probability of moving to the optimum if a particle is at one of the
  two attractors whereas for \onepso the probability of moving towards
  the optimum if the particle is located at the attractor tends to zero for
  large $n$.

  An analysis of experiments of \dpso with small number of particles and
  \onepso applied to the sorting problem and \onemax revealed only a small
  increase in the optimization time of \dpso compared to \onepso. This increase
  is way smaller than the factor $P$.
  For some parameter constellations also a significant decrease of the
  optimization time of \dpso compared to \onepso is achieved.

\subsection{Lower Bounds for Pseudo-Boolean Functions}
\label{sec:pseudoboolean:lb}
Also for general pseudo-Boolean functions $f:\lbrace 0,1\rbrace^n\rightarrow \R$ we can prove lower bounds on the expected optimization time.

\begin{theorem}
  \label{thm:pseudoboolean:const}
  If $P=\Theta(1)$, where $P$ is the number of particles, then the expected
  optimization time of \dpso optimizing pseudo-Boolean functions ($\lbrace
  0,1\rbrace^n\rightarrow\R$) with a unique optimal position is in
  $\Omega(n\log(n))$.
\end{theorem}
\begin{proof}
    If there are $P=\Theta(1)$ particles, then in expectation there are $n/2^P=\Omega(n)$ bits such that there is no particle where this bit of the optimal position equals the corresponding bit of the initial position.
    The expected optimization time is therefore bounded by the time that each such bit is flipped in a random move at least once.
    This subproblem corresponds to a coupon collectors problem and therefore we have the claimed lower bound of $\Omega(n\log(n))$.
\end{proof}

For larger values of $P$ we obtain an even higher lower bound by the following theorem.

\begin{theorem}
  \label{thm:pseudoboolean:poly}
  If $P=O(n^k)$, where $P$ is the number of particles and $k$ is an arbitrary non-negative
  real value, then the expected optimization time of \dpso optimizing
  pseudo-Boolean functions ($\lbrace 0,1\rbrace^n\rightarrow\R$) with a unique
  optimal position is in $\Omega(n\cdot P)$.
\end{theorem}
\begin{proof}
  To bound the probability to be at least some distance apart from the attractor after initialization we can use Chernoff bounds.
  For a fixed particle we can define
  $$Y_i=\begin{cases}
    1 & \text{\begin{tabular}{@{}l@{}}if the $i$th bit of the initial position differs from\\the corresponding bit of the unique optimal position\end{tabular}}\\
    0 & \text{otherwise}\enspace.
  \end{cases}
    $$
  Therefore $Y=\sum_{i=1}^n Y_i$ is exactly the initial distance of the fixed particle to the unique optimal position.
  For each $i$ we have that $\Pr[Y_i=1]=\frac{1}{2}$ and $\E[Y]=\frac{n}{2}$.
  By Chernoff bounds we obtain the lower bound
  \begin{align*}
  \Pr\left[Y > \frac{n}{4}\right]
  &=1-\Pr\left[Y\leq\left(1-\frac{1}{2}\right)\E[Y]\right]
  \geq 1-\exp\left(-\frac{\left(\frac{1}{2}\right)^2\cdot \frac{n}{2}}{2}\right)
  \\&=1-\exp\left( -\frac{n}{16} \right)\enspace.
\end{align*}
  The probability that the initial position of all $P$ particles have distance
  at least $\frac{n}{4}$ to the unique optimal position is the $P$th power of
  this probability and can be bounded from below for large $n$ by
$$
\left(1-\exp\left( -\frac{n}{16} \right)\right)^P
\geq 1-P \cdot \exp\left( -\frac{n}{16} \right)
\overset{n\geq 16\ln(2P)}{\geq}  \frac{1}{2}\enspace.
$$
Please note that one can choose such an $n$ as $P=O(n^k)$ and
$16\cdot \ln(2\cdot{\rm poly}(n))=o(n)$ for any polynomial.
If the distance of the positions of all particles is at least $\frac{n}{4}$
then it takes at least $\frac{n}{4}$ iterations until the optimal position can
be reached as the distance can change only by one in each iteration. For each
iteration $P$ evaluations of the objective function are performed.
Therefore we have at least $\frac{n\cdot P}{4}$ objective function evaluations
with probability at least $\frac{1}{2}$ for large $n$ which results in the
claimed optimization time of $\Omega(n\cdot P)$.
\end{proof}

This means if we choose , e.\,g., $P=n^{10}$ we would have at least
$\Omega(n^{11})$ function evaluations in expectation.

\section{Conclusion}
\label{sec:conclusion}

We propose a simple and general adaptation of the PSO algorithm for a broad class of discrete optimization
problems. For one particle, we provide upper and lower bounds on its 
expected optimization time for the sorting problem and \onemax and generalize
the upper bounds to \dpso with arbitrary number of particles and we also prove
lower bounds of \dpso optimizing pseudo-Boolean functions. Depending on the
parameter $c$, which is the probability of moving towards the
attractor, the expected optimization time may be polynomial ($c \geq 1/2$) and
exponential ($c < 1/2$), resp. The cornerstone of our analysis are $\Theta$-bounds
on the expected time it takes until the PSO returns to the attractor. Our analysis also provides the variance of this value.
We analyze Markov chains and provide tools to evaluate
expected return times for certain classes of transition probabilities.
Additionally we establish a useful general property of indistinguishability
of a Markov process for obtaining lower
bounds on the expected first hitting time of a special state.
Application of the presented tools on other Markov chains,
often appearing in the analysis of randomized algorithms, would obviously be possible.

For future work, it would be interesting to see if the upper and lower bounds
on the expected optimization time for \onemax given in Theorems~\ref{thm:onepso:onemax}
and~\ref{thm:onepso:onemaxlb} are valid for any linear function $f : \{0,1\}^n
\rightarrow \real$, $f(x_1,x_2,\ldots,x_n) = \sum_i w_i x_i$. Furthermore, we
conjecture that the upper bounds on the sorting problem for $c=0$ is $n!$ and
that the other proved upper bounds on the sorting problem are tight. 
Another direction for future work is to apply our technical tools to other
meta-heuristics. In particular, our tools may be useful in the analysis of
``non-elitist'' meta-heuristics, for instance the 
Strong Selection Weak Mutation (SSWM)
evolutionary regime introduced in~\cite{Gil:83} as an example of non-elitist
algorithm.

Finally, it would be interesting to determine the return time to the state
$S_0$ in a more general Markov model $\model((p_i)_{1 \leq i \leq n})$, where
$p_i = 1/2 + z(i, n)$ such that $z(i, n) = \operatorname{poly}(i) /
\operatorname{poly}(n)$, where the degrees of the polynomials differ, and
$z(i,n)$ is non-decreasing for $1 \leq i \leq n$.
This would generalize Theorems~\ref{thm:h1_lin_theta}
and~\ref{thm:h1_quad_theta}, and shed some light on the relation between
$z(i,n)$ and the return time to state $S_0$. Here, we conjecture that for
$z(i,n)$ as defined above the return time is in $\operatorname{poly}(n)$.
Finally a proof for the claimed upper bound of $O(n!)$ on the expected time to
reach a specified permutation in the graph of permutations by an actual random
walk searching for the optimum would be beneficial. 
To the best of our knowledge no proof exists so far.

\section*{Acknowledgements}
We would like to thank the anonymous referees for valuable remarks.

\bibliographystyle{alpha}
\bibliography{literature}   

\newcommand{\GenerateBio}[3]{
\noindent
\begin{minipage}{0.23\textwidth}
\centering\includegraphics[width=1in,height=1.25in,clip,keepaspectratio]{#1}
\end{minipage}
\hfill
\begin{minipage}{0.765\textwidth}
\textbf{#2}
#3
\end{minipage}\\[\baselineskip]%
}

\appendix

\section{Evidence for Conjecture~\ref{conject:random_walk}}
\label{subsec:randomwalkconjecture}
	In this section we provide computational evidence for
	Conjecture~\ref{conject:random_walk}. To this end we compute exact
	values for $T_{\rm sort}(n)$ for $n \leq 40$.  For the calculation of
	$T_{\rm sort}(n)$ for small $n$ a system of linear equations similar to
	Equation~\ref{eq:recmarkov} is used. Let $\tau_0$ be the optimal
	permutation (say, the identity) then
  \begin{equation}
  \begin{aligned}
	h_{\tau_0} &= 0 , &\\
	h_{\tau} &= 1 + \sum_{\nu\in\neighborhood(\tau)}\frac{1}{\vert \neighborhood(\tau)\vert} h_{\nu}, & \tau \text{ a permutation}\enspace.
    \label{eq:htau}
  \end{aligned}
  \end{equation}
  This simple approach works only for very small $n$ since one variable for
  each permutation is used. Our results are based on the following insight: For
  each permutation $\tau$ we examine the permutation $\nu$ such that
  $\tau\circ\nu=\tau_0$. Since the value $h_\tau$ is equal for all permutations
  with the same cycle lengths of the cycle decomposition of $\nu$ the number of
  variables in the system of linear equations can be reduced to the number of
  integer partitions of $n$, where $n$ is the number of items to sort. Hence,
  for $n = 40$, we have reduced the number of variables from $40!$ to $37\,338$,
  which is a manageable number. $T_{\rm sort}(n)$ is then just a linear
  combination of the calculated values.

  Side note: The transition probabilities as well as the system of linear
  equations can be represented by a matrix. For elitist algorithms the matrix
  can be transformed to an upper-triangular matrix and therefore it is much
  easier to analyze them. Such a transformation to an upper-triangular matrix
  is not possible in our situation.

  \begin{figure}[htb]
\centering
\begin{tikzpicture}[
auto
]
\tikzset{>={Stealth[length=2.5mm]}}
\node[ellipse, draw] (v1111) {$(1,1,1,1)$};
\node[ellipse, draw, right = of v1111] (v211) {$(2,1,1)$};
\node[right = of v211] (vhelpa) {};
\node[right = of vhelpa] (vhelpb) {};
\node[right = of vhelpb] (vhelpc) {};
\node[ellipse, draw, above = of vhelpb] (v31) {$(3,1)$};
\node[ellipse, draw, below = of vhelpb] (v22) {$(2,2)$};
\node[ellipse, draw, right = of vhelpc] (v4) {$(4)$};
\path[]
(v1111) edge [->,bend right] node [below] {$1$} (v211)
(v211)  edge [->,bend right] node [above] {$\frac{1}{6}$} (v1111)
(v211)  edge [->,out=10,in=-130] node [above left] {$\frac{2}{3}$} (v31)
(v211)  edge [->,bend right] node [below left] {$\frac{1}{6}$} (v22)
(v22)   edge [->,out=130,in=-10] node [below left] {$\frac{1}{3}$} (v211)
(v22)   edge [->,bend right] node [below right] {$\frac{2}{3}$} (v4)
(v31)   edge [->,bend right] node [above left] {$\frac{1}{2}$} (v211)
(v31)   edge [->,out=-50,in=-190] node [above right] {$\frac{1}{2}$} (v4)
(v4)    edge [->,out=190,in=50] node [below right] {$\frac{1}{3}$} (v22)
(v4)    edge [->,bend right] node [above right] {$\frac{2}{3}$} (v31)
;
\end{tikzpicture}
\caption{Search space for the problem of sorting four items by transpositions. The states are partitioned by their cycle lengths.}
\label{fig:cycles4}
\end{figure}
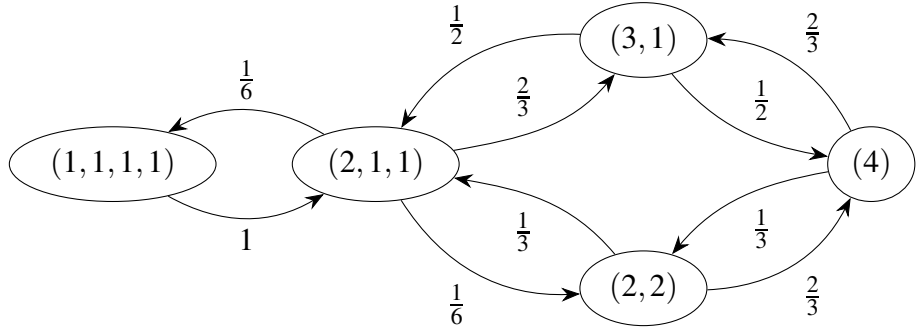
In Figure~\ref{fig:cycles4} we can see the search space of the sorting problem
for four items partitioned by their cycle lengths. This results
in states which are represented by the integer partitions of $n=4$. The
complete search space with all permutations with $n=4$ items is already
visualized in Figure~\ref{fig:permutations}. In Figure~\ref{fig:cycles4} each state is labeled by the
cycle lengths of the current permutation and can represent different
permutations.  $(1,1,1,1)$ is only a single permutation - the identity
permutation $1234$ - where each cycle is a singleton cycle - a cycle with
length one. All neighboring permutations of $(1,1,1,1)$ contain two swapped
items and two items which stay at their position. Therefore the cycle lengths
are $(2,1,1)$ and there are six permutations with these cycle lengths.  The
neighbors of permutations with cycle lengths $(2,1,1)$ can have cycle lengths
$(1,1,1,1)$, $(3,1)$ or $(2,2)$.  Out of the six possible exchange operations
one splits the cycle of length two into two cycles of length one, one exchange
operation merges the two singleton cycles to one cycle of length two getting
two cycles of length two and the remaining four exchange operations merge the
cycle of length two with a singleton cycle getting cycle lengths $(3,1)$. The
respective transition probabilities for a random walk are also visualized in
Figure~\ref{fig:cycles4}. Furthermore, there are three permutations with cycle
lengths $(2,2)$ - the permutations $2143$, $3412$ and $4321$. The remaining
eight permutations of the third column in Figure~\ref{fig:permutations} have
cycle lengths $(3,1)$ and the six permutations in the last column in
Figure~\ref{fig:permutations} have only a single cycle of length four.
The values of $h_\tau$ satisfying Equation~\ref{eq:htau} are $h_{(1,1,1,1)}=0$,
$h_{2,1,1}=23$, $h_{(2,2)}=27$, $h_{(3,1)}=\frac{105}{4}$,
$h_{(4)}=\frac{55}{2}$ and then we have
$T_{\rm sort}(4)=\frac{1}{4!}\cdot(h_{(1,1,1,1)}+6\cdot h_{2,1,1}+3\cdot h_{(2,2)}+6\cdot h_{(3,1)}+6\cdot h_{(4)})=\frac{99}{4}=4!+\frac{3}{4}$.

  Please note that this value does not rely on experiments. Instead the
  evaluations result in the exact expected optimization time.

\begin{figure}[htb]
  \centering
  \input{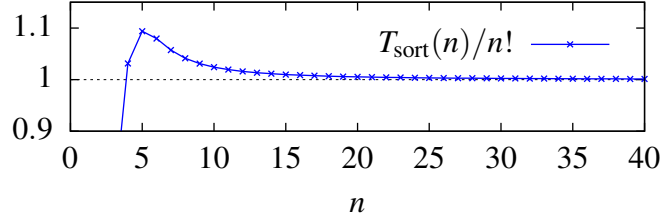}
  \caption{Expected time it takes for the random walk to reach the sorted
  sequence of $n$ items in the graph of permutations divided by $n$ factorial.}
  \label{fig:random_walk_sorting}
\end{figure}
  In Figure~\ref{fig:random_walk_sorting} the exact value of $T_{\rm sort}(n)$
  divided by $n!$ tends to one. Proving that this holds for large $n$ implies that $T_{\rm sort}(n)
  =\THETA(n!)$ and also Conjecture~\ref{conject:random_walk}.

\section{Evidence for Conjecture~\ref{con:sorting:lb}}
\label{subsec:approxOptTime}
For the sorting problem, the lower and upper bound on the return time to the
attractor have the following gaps. For $0<c<1/2$ the expected number of
iterations to return to the attractor vary from $\Omega^{*}(\alpha(c)^n)$ to
$O^*\left( \left( (1-c)/{c} \right)^n \right)$ and for $c=1/2$ these values
vary from $\Omega(n^{2/3})$ to $O(n)$.
We provide a simplified Markov model based on an averaging argument. We
conjecture that the simplified model and the actual model are asymptotically
equivalent.

To improve the understanding of the search space of permutations we will
approximate the improvement probabilities to obtain approximations of the expected return time to
the attractor. For this purpose, instead of using upper and lower bounds on the
probability to move towards the attractor, we use the average value.
If the probability to be in a specific permutation equals the probability to be
in any other permutation with the same distance to the attractor then this
approximation would also result in the exact values.
\begin{conjecture}
  Let $H_1$ be the expected number of iterations until the \onepso returns to the attractor $g$ if the current distance to the attractor is one while optimizing the sorting problem.
  Let $p_x$ be the probability to move from permutation $x$ to a permutation $y$ such that $\dist(x,g)=1+\dist(y,g)$.
  Let
  \[\hat p_i={\displaystyle\sum_{x\in X_i} p_x}/{\vert X_i\vert}\]
  be the average probability to reduce the distance to the attractor.
Let $\hat H_1$ be the expected number of iterations in $\hat M=\model( (\hat p_i)_{1\leq i\leq n-1}))$ to move from state $S_1$ to $S_0$ (see Def.~\ref{defi:model}).
We conjecture that $H_1\sim\hat H_1$.
\end{conjecture}
To provide evidence we compute these average improvement probabilities and
compare the number of expected iterations.
\begin{theorem}
  \label{thm:average_improvement}
  The average improvement probability of moving towards the attractor while optimizing the sorting problem is
  $$
  \hat p_i=c+(1-c)\cdot\frac{\displaystyle\sum_{k=1}^{i+1}\left( \frac{k-1}{n-1}\cdot\frac{(n-1)!}{(n-k)!}\cdot\stirling{n-k}{n-i-1} \right)}{\displaystyle\stirling{n}{n-i}}\enspace,
  $$
  where $i$ is the distance to the attractor and $\stirling{n}{m}$ are the unsigned Stirling numbers of the first kind.
\end{theorem}
\begin{proof}
  The unsigned Stirling numbers of the first kind $\stirling{n}{m}$ represent
  the number of permutations of $n$ elements with exactly $m$ cycles which can easily be calculated by the recursive formula
  \begin{align*}
    \stirling{n}{m}=\begin{cases}
      1&\text{if }n=0\wedge m=0\\
      0&\text{if }n=0\wedge m\neq 0\\
      \stirling{n-1}{m-1}+(n-1)\cdot \stirling{n-1}{m}&\text{if }n>0
    \end{cases}
  \end{align*}
  W.l.o.g let the attractor be the identity permutation. Then the attractor has $n$ singleton cycles.
  An increase of the distance to the attractor by one is equivalent to a decrease of the number of cycles by one.
  Therefore a permutation with distance $i$ to the attractor has exactly $n-i$ cycles.
  This means that the number of permutations with distance $i$ is $\stirling{n}{n-i}$.
  The probability that a fixed item is in a cycle of length $k$ among all permutations with distance $i$ from the attractor is 
  $$\frac{\binom{n-1}{k-1}\cdot(k-1)!\cdot\stirling{n-k}{n-i-1}}{\stirling{n}{n-i}}=\frac{\frac{(n-1)!}{(n-k)!}\cdot\stirling{n-k}{n-i-1}}{\stirling{n}{n-i}} \enspace.$$
  Choosing the remaining $i-1$ items in the cycle of length $k$ from the
  remaining $n-1$ items has $\binom{n-1}{k-1}$ options.  There are $(k-1)!$
  orderings of these items within the cycle. The remaining $n-k$ items have to
  be partitioned into $n-i-1$ cycles which results in another factor of
  $\stirling{n-k}{n-i-1}$ options. In combination with the first cycle of length
  $k$ a permutation with $n-i$ cycles is achieved.

  This probability does not change if we choose a random item instead of a fixed
  item.  Furthermore the probability of moving towards the attractor is
  determined by the probability that a cycle is split into two cycles which
  happens if two items of the same cycle are picked for an exchange. If the
  first picked item is in a cycle of length $k$ then the probability that the
  second item is in the same cycle is $\frac{k-1}{n-1}$. Summing up these
  probabilities over all possible cycle lengths for the first picked item
  results in the claimed result for $\hat p_i$, but there also the constant $c$
  of the \dpso comes into play which forces a move towards the attractor.
  Please note that the maximal cycle length at distance $i$ from the attractor
  is $i+1$ which explains the upper limit of the sum.
\end{proof}

By using these average probabilities of moving towards the attractor we obtain
a Markov chain where it is only possible to move to state $S_{i-1}$ or
$S_{i+1}$ from state $S_i$ (and not to any other state) in a single step. For
Markov chains with this property the return times can be computed as in
Section~\ref{subsec:const}.
The result helps us to estimate the expected return time to the attractor. If
$c$ is zero then the expected return time if we are at distance one to the
attractor is exactly $n!-1$ which is also obtained exactly by the model with
average probabilities.

\begin{remark}
  \label{rem:approx}
  \hfill
  \begin{enumerate}
    \item 
Assuming $T(n)=\gamma(c)^n\cdot f(n)$ where $f$ is  a
polynomial then
$$\lim_{n\rightarrow\infty}\frac{T(n)}{T(n-1)}=\gamma(c).$$
    \item 
Assuming $T(n)=f(n)$ where $f$ is again a polynomial then
$$\lim_{n\rightarrow\infty}\log_{\frac{n}{n-1}}\left( T(n)/T(n-1) \right)$$
is the maximal degree of $f(n)$.
  \end{enumerate}
\end{remark}
\begin{proof}
  Let $f(n)=a\cdot n^b+o(n^b)$, $b>0$.
  \begin{enumerate}
    \item 
	  Assuming $T(n)=\gamma(c)^n\cdot f(n)$ leads to
      \begin{align*}
  \lim_{n\rightarrow\infty}\frac{T(n)}{T(n-1)}
  &=\lim_{n\rightarrow\infty}\gamma(c)\frac{n^b+o(n^b)}{(n-1)^b+o((n-1)^b)}\\
  &=\lim_{n\rightarrow\infty}\gamma(c)\frac{(\frac{n}{n-1})^b+o(1)}{1+o(1)}
  =\gamma(c)
\end{align*}
    \item 
	  Assuming $T_{\rm sort}(n)=f(n)$ leads to
      \begin{align*}
  \lim_{n\rightarrow\infty}\log_{\frac{n}{n-1}}\left(\frac{T(n)}{T(n-1)}\right)
  &=\lim_{n\rightarrow\infty}\log_{\frac{n}{n-1}}\left(\frac{(\frac{n}{n-1})^b+o(1)}{1+o(1)}\right)\\
  &=\lim_{n\rightarrow\infty}\log_{\frac{n}{n-1}}\left(\left(\frac{n}{n-1}\right)^b\right)
  =b
\end{align*}
  \end{enumerate}
\end{proof}

Let
$$q_{ex}(n,c):=\frac{H_1}{H_1'}$$
where $H_1$ and $H_1'$ are the expected
return times to the attractor for the sorting problem on $n$ and $n-1$ items
respectively if the attractor has transposition distance one to the current
position (actual \textbf{ex}act model).

Let additionally
$$q_{av}(n,c):=\frac{H_1}{H_1'}$$
where $H_1$ and $H_1'$ are the
corresponding return times in the Markov model with \textbf{av}erage success probability
specified in Theorem~\ref{thm:average_improvement}.

\begin{figure}[htb]
  \centering
  \input{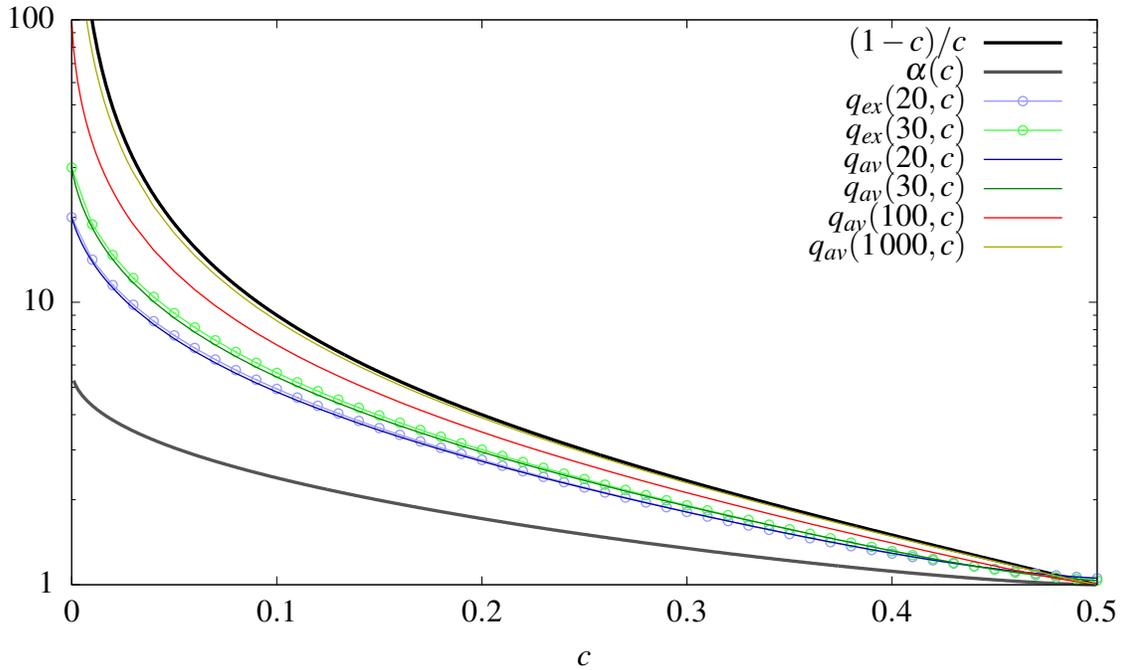}
  \caption{Quotients of return times $q_{ex}$ and $q_{av}$ for different values
  of $n$ and the upper and lower bound on the limit of $q_{ex}$ if $n$ tends to
infinity.}
  \label{fig:returntime_average}
\end{figure}

In Figure~\ref{fig:returntime_average} $q_{ex}$ and $q_{av}$ for different
values of $n$ and the upper bound $(1-c)/c$ and the lower bound $\alpha(c)$ on
the base of the exponential part of the return time for the sorting problem.

With the first part of Remark~\ref{rem:approx} we also notice that $(1-c)/c$ is also an upper
bound and $\alpha(c)$ is a lower bound on the limit of $q_{ex}$ if $n$ tends to
infinity as $q_{ex}$ tends to the actual base of the exponential part of the
expected return time.

As $q_{ex}$ has to be calculated by a system of linear equations where the
number of variables equals the number of integer partitions of $n$ (see
description after Conjecture~\ref{conject:random_walk}) we can not evaluate
$q_{ex}$ for large $n$. The values of $q_{av}$ are quite similar to the values
of $q_{ex}$ for corresponding $n$. For $c=0$ the values are exactly the same
and for $n=30$ the relative error is less than $0.04$.
Therefore we conjecture that the limit of $q_{av}$ if $n$ tends to infinity is
close or even equal to the limit of $q_{ex}$. But as we can see
in Figure~\ref{fig:returntime_average} the values of $q_{av}$ tend to the upper
bound of $(1-c)/c$. We omitted the graph of $q_{av}(10\,000,c)$ as it overlaps the
graph of the upper bound almost completely. So it is reasonable to conjecture
that the limit
of $q_{ex}$ is close to the upper bound $(1-c)/c$. If this is actually true
then for all runtime results the value of $\alpha(c)$ can be replaced by
$((1-c)/c)-\varepsilon$ for some small non-negative value $\varepsilon$ which
could probably be even zero.

\begin{figure}[htb]
  \centering
  \input{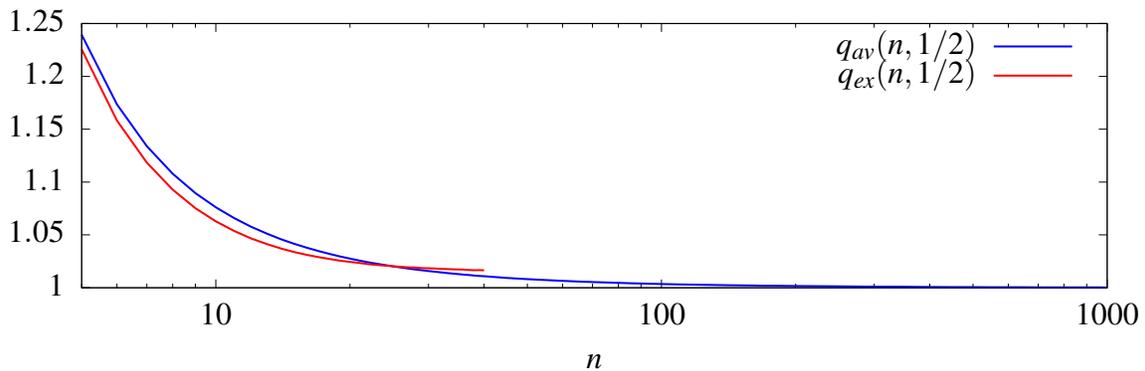}
  \caption{Quotients of return times $q_{ex}$ and $q_{av}$ for $c=1/2$.}
  \label{fig:returntime_average05}
\end{figure}

By using the second part of Remark~\ref{rem:approx} the limit of
$\log_{\frac{n}{n-1}}q_{ex}(n,1/2)$ if $n$ tends to infinity supplies us the exponent
of the largest monomial (probably omitting logarithmic factors) of the return
time if $c=1/2$.  In Figure~\ref{fig:returntime_average05} we can see the
quotients $q_{ex}(n,1/2)$ for n up to $40$ and $q_{av}(n,1/2)$ for even larger
values of $n$. Also here $q_{av}$ can be used as an approximation on $q_{ex}$
and it is reasonable to assume that the limit is one. Please note that
Theorem~\ref{thm:onepso:sorting} tells us that the limit
$\lim_{n\rightarrow\infty}q_{ex}(n,1/2)\leq 1$.
Similarly to the exponential case with $c<1/2$ we conjecture that the actual
expected value $H_1$ is close to the proposed upper bound on $H_1$ described in
the proof of Theorem~\ref{thm:onepso:sorting}.

Using these results for lower bounds on the expected optimization time we would have
  \[
	T_{\rm sort}(n) = 
	\begin{cases}
	  \Omega(n^2)	&	\text{if $c \in (\frac{1}{2}, 1]$} \\
	  \Omega(n^{3})	&	\text{if $c = \frac{1}{2}$} \\
      \Omega \left(\left(\frac{1-c}{c}\right)^n\cdot n^2\right)	&	\text{if $c \in (0, \frac{1}{2})$}
	\end{cases}
  \]
  as specified in Conjecture~\ref{con:sorting:lb}.
  These bounds are only a factor of $\log(n)$ apart from the upper bounds
  specified in Theorem~\ref{thm:onepso:sorting}.

  Please note that the results in this section do not rely on
  experiments. Instead, exact values for the expectation are computed.
  Nevertheless, we have done some experiments with small values of $n$ and the
  measured optimization times comply with the evaluated exact optimization
  times.
  Especially for values $c$ close to
  zero, optimization times of up to $n!$ are claimed which can not be
  confirmed in reasonable time for larger values of $n$.

\end{document}